\documentclass[11pt]{article}
\pdfmapfile{=jquant_fonts.map}
\usepackage[journal=jquant,stage=submission,logging=stdout]{jsty3-author}
\usepackage{mathtools}
\usepackage{amsthm}
\usepackage{braket}
\usepackage{bm}

\providecommand{\openone}{\mathbb{1}}
\providecommand{\onlinecite}[1]{~\cite{#1}}

\newcommand{\ri}{\mathrm{i}}
\newcommand{\I}{\openone}
\newcommand{\E}{\mathbb{E}}
\newcommand{\C}{\mathbb{C}}

\newcommand{\n}{\mathbf{n}}

\newcommand{\Hil}{\mathcal H}
\newcommand{\B}{\mathcal B}
\newcommand{\Tr}{\operatorname{Tr}}
\newcommand{\Var}{\operatorname{Var}}

\newcommand{\dd}{\mathrm d}
\newcommand{\sphere}{\mathcal S^2}
\newcommand{\HS}{\mathrm{HS}}

\newcommand{\red}[1]{{\color{black}#1}}

\theoremstyle{plain}
\newtheorem{theorem}{Theorem}
\newtheorem{proposition}[theorem]{Proposition}

\theoremstyle{definition}
\newtheorem{definition}[theorem]{Definition}

\theoremstyle{remark}
\newtheorem{remark}[theorem]{Remark}

\title{Spin-coherent quantum designs}
\author*{a,b}{Marcin~Rudzi\'nski}{m.rudzinski@doctoral.uj.edu.pl}{0000-0002-6638-3978}
\author{c}{Aaron~Z.~Goldberg}{}{0000-0002-3301-7672}
\author{d}{Andrei~B.~Klimov}{}{0000-0001-8493-721X}
\author{e,f,g}{Luis~L.~S\'{a}nchez-Soto}{}{0000-0002-7441-8632}
\author{a,h}{Karol~\.Zyczkowski}{}{0000-0002-0653-3639}

\affiliation{a}{Faculty of Physics, Astronomy and Applied Computer Science, Jagiellonian University, 30-348 Krak{\'o}w, Poland}
\affiliation{b}{Doctoral School of Exact and Natural Sciences, Jagiellonian University, 30-348 Krak{\'o}w, Poland}
\affiliation{c}{National Research Council of Canada, Ottawa, Ontario K1N 5A2, Canada}
\affiliation{d}{Departamento de F\'{\i}sica, Universidad de Guadalajara, 44420~Guadalajara, Jalisco, Mexico}
\affiliation{e}{Departamento de \'Optica, Facultad de F\'{\i}sica, Universidad Complutense, 28040~Madrid, Spain}
\affiliation{f}{Max-Planck-Institut f\"ur die Physik des Lichts, 91058~Erlangen, Germany}
\affiliation{g}{Institute for Quantum Studies, Chapman University, Orange, CA 92866, USA}
\affiliation{h}{Center for Theoretical Physics, Polish Academy of Sciences, 02-668 Warszawa, Poland}

\abstract{
	Coherent states bridge the gap between quantum and classical physics, but their overcomplete and nonorthogonal nature makes it difficult to identify 
	\red{finite sets that represent spin observables exactly.}
	Finite spin-coherent tomography and discrete coherent-state operator bases 
    \red{already exist. Here we address a more specific problem: preserving the canonical contravariant-symbol representation.} 
	\red{For a spin-\(J\) system, the canonical finite coherent-state formula reconstructs every operator in the rank-\(S\) irreducible tensor sector exactly if and only if the sampling directions form a spherical \((2J+S)\)-design and the corresponding coherent-state projectors then form a spin-coherent quantum design.}
	\red{We also give an explicit Gauss–Legendre construction which leads to weighted spherical designs. These results provide a systematic way to determine spin-operator moments from finite coherent-state POVMs without full state tomography, with practical protocols for polarimetry and magnetometry.}
	\red{ Presented construction can be extended to other compact symmetry groups.}
}

\let\linenumbers\nolinenumbers
\nolinenumbers
\begin{document}

	
	\section{Introduction}
	\label{sec:introduction}
	
	Coherent states occupy a privileged place in physics: they track classical dynamics remarkably closely, which makes them a natural language of simulation and measurement. Yet this closeness to classicality comes at a price: they are not orthogonal~\cite{Perelomov:1986aa,Gazeau:2009aa,Robert:2021aa,Kam:2023aa}. This raises the question: for a given task, which coherent states does one actually need to know?
	
	Consider completeness relations, one of the defining properties of any family of coherent states~\cite{Klauder:1985aa}. Standard quantum theory dictates  that an infinite continuum of them must be summed to resolve the identity; how excessive is that? How many coherent states are truly needed to represent a given observable, and how does that number grow as the operator complexity increases?
	Because coherent states are nonorthogonal, their phase-space representations~\cite{Schroeck:1996fv,Schleich:2001hc,QMPS:2005mi}
	generally involve nonorthogonal operator frames and dual symbols rather than
	orthogonal projective measurements. Nevertheless, suitably weighted
	coherent-state projectors can form a positive operator-valued measure.
	This raises the same question in the setting of quantum tomography: how many coherent states must a state be compared against to fully determine its informational content, or, equivalently, which values of a quasiprobability distribution suffice to represent it? A closely related line of work asks how tightly an operator can be constrained once it is measured on only a restricted set of coherent states~\cite{Janssen:1981aa}.
	
	Framed more broadly, representing a continuous function through a finite set of discrete evaluations is the subject of cubature~\cite{Ueberhuber:1997aa} and design theory~\cite{Delsarte:1977aa,Goethals:1981aa,Colbourn:2010aa}. Lifted to the quantum setting, this idea becomes coherent quantum designs that mimic operators by discrete sets of coherent states. 
	
	These questions have all been addressed in the realm of canonical coherent states, where the phase space is the two-dimensional plane and the symmetry algebra is Heisenberg-Weyl.  There, the identity can be resolved by coherent states placed on an infinite lattice whose unit cell has area $h$~\cite{Neumann:1932aa,Perelomov:1971aa,Bargmann:1971aa,Bacry:1975aa,Boon:1983aa}--the very same lattice that is ideal for  the Gottesman-Kitaev-Preskill encoding of qubits in oscillators~\cite{Gottesman:2001aa} and its numerous applications~\cite{Menicucci:2014aa,Duivenvoorden:2017aa,Fukui:2021aa,Conrad:2022aa}. Bosonic tomography can thus be restricted to sampling this lattice, reducing the overhead of heterodyne detection. The lesson is a subtle one: representing an arbitrary state or operator on the infinite plane still demands infinitely many coherent states, but not, crucially, a continuum of them.

	Here we develop a rank-resolved framework for SU(2). The spin
	coherent states \(|\n\rangle\) (also known as atomic or Bloch coherent
	states)~\cite{Atkins:1971aa,Radcliffe:1971aa,Arecchi:1972aa}  are widely regarded as the least quantum states of a spin system, not least because they saturate a variety of uncertainty relations~\cite{Kitagawa:1993aa,Dammeier:2015aa,Shabbir:2016aa,
		Goldberg:2020aa}. Like their canonical counterparts, they satisfy a
	continuous resolution of the identity. Finite reconstruction from spin
	coherent states has substantial precedent: coherent-state probabilities
	have been used to reconstruct pure spin states~\cite{AmietWeigert1999},
	mixed states and arbitrary spin operators from finitely many values of
	their \(Q\) symbols~\cite{AmietWeigert2000}, and discrete spin phase-space
	representations~\cite{Weigert1999}. 
	A full discrete Moyal-type calculus was subsequently developed using
	\((2J+1)^2\) sphere points and a constellation-dependent dual operator
	basis.~\cite{HeissWeigert2000}
	Generic families of
	\((2J+1)^2\) coherent-state projectors can also form operator
	bases~\cite{Weigert2004}, while finite coherent-state measurements based
	on spherical quadratures, including a Gauss--Legendre product
	construction, were developed in Ref.~\onlinecite{IblisdirRoland2006}.
	
	Our question is more specific: when can the continuous coherent-state
	representation be replaced by a finite sum, using the canonical
	contravariant symbols, uniformly over a complete irreducible tensor
	sector? We show that exact reconstruction of the entire rank-\(S\)
	sector is equivalent to the spherical \((2J+S)\)-design condition. 
	We refer to the coherent states associated with these constructions as \emph{spin-coherent quantum designs}.
    \red{This term was introduced by Hu and Tomamichel~\cite{HuTomamichel2024} to calculate the uniform orbit average of the $\ell$-fold tensor power of a coherent-state projector.}
	Unlike quantum $t$-designs, which average functions over the whole of Hilbert space~\cite{Ambainis:2007aa}, spin-coherent designs involve only states that are simple to prepare, making them considerably more practical.
	
	Spin systems  have an additional degree of freedom, the spin $J$, or equivalently the dimension $2J+1$, which lets us track how the number $M$ of coherent states needed to represent an operator scales with dimension. There is no unique way of distributing $M$ points evenly over a sphere~\cite{Saff:1997aa}, which is precisely what makes spherical designs so intriguing  and leads to problems such as sphere packing~\cite{Conway:1996ys,Viazovska:2017aa}.
	Consequently, unlike the planar case, the resulting lattice on the sphere is not built from cells of area exactly $h$; nonetheless, we find that the average cell area still decreases as $J^{-2}$, and hence plausibly as $h^2$. 
	For observables supported on tensor ranks \(0\le S\le R\le2J\), the required
	design degree is \(2J+R\).
	
	We put this machinery to practical use by estimating moments of spin operators, a standard task in physical applications, including  polarimetry~\cite{Berry:1977aa,Azzam:1985aa,Goldberg:2020ac} and magnetometry~\cite{Wasilewski:2010aa,Behbood:2013aa}. Our spin-coherent quantum designs translate directly into experimentally friendly protocols for reconstructing these moments, including in the realistic regime of finite measurement statistics. Quantum designs have already been used to simplify the analysis of mutually unbiased bases~\cite{Renes:2004aa,Klappenecker:2005aa} and in tasks such as shadow tomography~\cite{Cieslinski:2024aa};  so we expect that spin-coherent designs will streamline them even further (unlike the case of Gaussian states for which such designs are absent~\cite{Blume:2014aa}).  It is pleasing that these algebraic and geometrical studies have practical application and we speculate that a similar analysis can be carried out in a variety of phase spaces. 
	
	\section{Spin coherent states and spherical designs}
	\label{sec:coherent-designs}
	
	Let $\Hil_{J} \cong \mathbb{C}^{2J+1}$ be the carrier space of the irreducible spin-\(J\) representation of  SU(2), of dimension $d = 2J+1$,  and let $J_x, J_y, J_z$ be the angular-momentum operators satisfying (with \(\hbar=1\) throughout)
	\begin{equation}
		[J_x,J_y] = \ri   J_z  \, ,
	\end{equation}
	and cyclic permutations.   The space \(\mathcal H_J \)  is spanned by the standard basis $\{ |J m\rangle \mid m= -J, \ldots, J\}$  of simultaneous eigenstates  of $\mathbf{J}^{2}$ and ${J}_{z}$. 
	
	Spin coherent states are obtained by rotating the highest-weight state 
	\(|JJ\rangle\)~\cite{Perelomov:1986aa}.  For a unit vector \(\n\in\sphere\), let \(R_{\n} \in SO(3)\) be the rotation sending the north pole \(\mathbf e_z\) to \(\n\).  The spin coherent state pointing along  \(\n\) is 
	\begin{equation}
		|\n\rangle_J	= 	U_J (R_{\n}) |J J \rangle \, ,
	\end{equation}
	where $U_J$ is the irreducible spin-$J$ representation. 
	
	The state vector depends on the phase convention used to choose \(R_{\n}\), but the rank-one projector
	\begin{equation}
		P^{(J)}(\n) =	|\n\rangle \langle \n|_J
	\end{equation}
	is unambiguous and  transforms covariantly
	\begin{equation}
		U_J (R) \, P(\n) \, U_J(R)^\dagger = P(R\n), \qquad R \in SO(3) \, ,
	\end{equation}
	where we drop the superscript \(J \) once the spin is fixed.
	
	Geometrically, the coherent-state orbit SU(2)/U(1)$\simeq \sphere$ is the classical phase space of a spin, the finite-dimensional counterpart of the phase plane for canonical coherent states. These states are overcomplete and resolve the identity
	\begin{equation}
		\I_{2J+1} =
		\frac{2J+1}{4\pi} \int_{\sphere} P(\n)\, \dd \Omega \, ,
		\label{eq:continuous-coherent-identity}
	\end{equation}
	where $d\Omega$ is the invariant measure on $\sphere$. The right-hand side commutes with every \(U_J (R)\) by covariance, so by irreducibility and Schur's lemma it must be proportional to the identity; the normalization follows by taking the trace.
	
	The main question of this paper is when the continuous integral \eqref{eq:continuous-coherent-identity}, and
	more generally coherent-state integral representations of spin operators, can be replaced by finite sums. The relevant finite point sets are spherical designs, introduced by Delsarte, Goethals, and Seidel~\cite{Delsarte:1977aa} (see also the surveys \onlinecite{Bannai:2009aa,Womersley:2018aa}); their existence at every design strength follows from the averaging-set theorem~\cite{Seymour:1984aa}, and the best known and near-optimal designs on \(\sphere\) have been tabulated numerically by Hardin and Sloane~\cite{Hardin:1996aa,SloaneRepository}.

	\begin{definition}
		A finite set $\{\n_k\}_{k=1}^M\subset\sphere$ is a spherical $t$-design if it reproduces the uniform average of every spherical polynomial \(f\) of degree at most \(t\)
		\begin{equation}
			\frac1M\sum_{k=1}^M f(\n_k)=\frac1{4\pi}\int_{\sphere}f(\n)\,\dd\Omega \, .
			\label{eq:spherical-design-definition}
		\end{equation}
		Equivalently,
		\begin{equation}
			\sum_{k=1}^M Y_{Kq }(\n_k)=0,
			\qquad 1\le K \le t,
			\qquad -K \le q \le K \, .
			\label{eq:design-harmonic-condition}
		\end{equation}
	\end{definition}
	
	The equivalence between \eqref{eq:spherical-design-definition} and \eqref{eq:design-harmonic-condition} follows from the decomposition of spherical polynomials into spherical harmonics~\cite{Sloan:2009aa}.  The \(K=0\) harmonic is constant and is fixed by normalization, while exact averaging of all nonconstant harmonics up to degree \(t\) is precisely the condition that all moments up to degree \(t\) agree with the uniform sphere average.
	
	Let \(M_t\) denote the minimal number of points in a spherical \(t\)-design on \(\sphere\). A counting argument gives the lower bound \(M_t=\Omega(t^2)\), and a matching upper bound, \(M_t=\Theta(t^2)\), was established in Ref.~\onlinecite{Bondarenko:2013aa} (where $\Omega$ and $\Theta$ refer to the standard Bachmann-Landau notation for growth rates). \red{Both are illustrated in Fig.~\ref{fig:spherical-design-scaling}.} This quadratic scaling matters for spin systems: the coherent-state projector \(P(\n)\) has harmonic degree at most \(2J\), so the design strength needed for identity resolution grows linearly with \(J\), while the number of design points grows quadratically.
	
	\begin{figure}[t]
		\centering
		\includegraphics[width=0.62\linewidth]{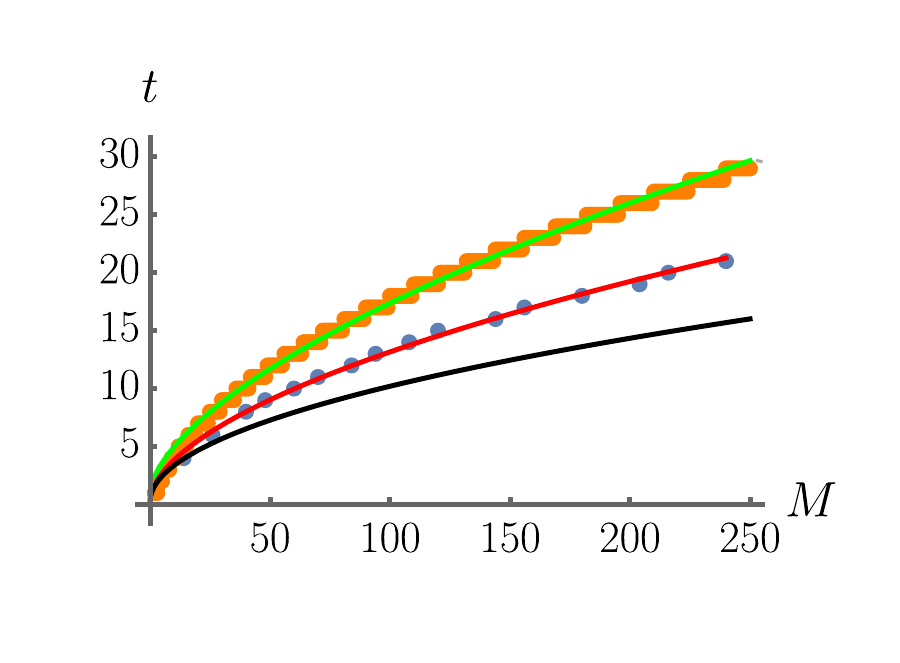}
		\caption{Scaling of known spherical designs on \(S^2\).  Blue points show the degree \(t\) of the design with the smallest known number \(M\) of points~\cite{SloaneRepository}.  Orange points show the lower bound~\cite{Delsarte:1977aa}, plotted as a constraint on the minimal number of points \(M_t\).  The red and green 	curves are empirical fits \(t\propto M^{\alpha}\) to the tabulated data and to the lower-bound points. \(\alpha\approx 0.56\) and \(\alpha\approx 0.61\), respectively.  The black curve shows the square-root approximation, \(t\propto\sqrt{M}\), illustrating the asymptotic scaling \(M_t=\Theta(t^2)\). }
		\label{fig:spherical-design-scaling}
	\end{figure}

	\begin{definition}
		Given a spherical \(t\)-design \(X=\{\n_k\}_{k=1}^{M}\subset \sphere\) and a fixed spin \(J\), the associated spin-coherent design is the image of \(X\) under the coherent-state embedding \(\n\mapsto P(\n)\),
		\begin{equation}
			\mathcal{P}_J (X) = \left \{ 	P(\n_k)= |\n_k \rangle  \langle \n_k| \right\}_{k=1}^{M} 	\subset \mathcal{B}(\Hil_J) \, ,
		\end{equation}
		where $\mathcal{B}(\Hil_J)  $  denotes the algebra of bounded linear operators acting on $\Hil_J$.
	\end{definition}
	
	\begin{remark}
		A spin-coherent design should not be confused with a complex projective \(t\)-design: it is a spherical \(t\)-design transported to the coherent-state orbit SU(2)/U(1)$\simeq \sphere \subset \mathbb{C}\mathbf{P}^{2J}$ and the design property is imposed only along this lower-dimensional orbit. The two notions coincide only for \(J=1/2\), where the orbit fills the entire projective space  $\mathbb{C}\mathbf{P}^{1}$.
	\end{remark}
	
	Under the adjoint action of the SU(2) group, the operator space of the spin-$J$ representation decomposes as
	\begin{equation}
		\B(\Hil_J)=\bigoplus_{K=0}^{2J}\mathcal V_K,
		\qquad
		\mathcal{V}_K=\operatorname{span}\{T^{(J)}_{Kq} \mid q=-K,\ldots,K\} \, ,
		\label{eq:operator-sector-decomposition}
	\end{equation}
	where $T^{(J)}_{Kq}$ are irreducible tensor operators, defined as~\cite{Fano:1959ly,Blum:1981rb}
	\begin{equation}
		T^{(J)}_{Kq} = \sqrt{\frac{2K+1}{2J+1}}
		\sum_{m,m'=-J}^{J}
		C^{\,J m'}_{Jm,\,Kq} 	|Jm'\rangle\langle Jm| ,
		\label{eq:app-TLM-definition}
	\end{equation}
	with \(C^{\,J m'}_{Jm,\,K q}\equiv \langle J,m;K,q|J,m'\rangle\) being a Clebsch--Gordan coefficient~\cite{Varshalovich:1988aa} that couples a spin $J$ and a spin $K$ ($0 \leq K \le 2J$) to a total spin $J$.  The coefficient vanishes unless \(q = m' -m\). These tensors form an orthonormal basis and have the right transformation properties under SU(2).
	
	In this basis, the coherent-state projector admits the multipole expansion
	\begin{equation}
		P(\n) = \sum_{K=0}^{2J} \sum_{q=-K}^{K}
		a_{JK} \, Y_{Kq}^\ast (\n)T^{(J)}_{Kq} \, ,
		\label{eq:projector-expansion}
	\end{equation}
	with
	\begin{equation}
		a_{JK} =	\sqrt{\frac{4\pi}{2J+1}}\,	C^{\,J J}_{JJ,\,K 0} = (2J)!\sqrt{\frac{4\pi}{(2J-K)!(2J+K+1)!}} \, .
		\label{eq:ajL-main}
	\end{equation}
	This multipole structure links spin-coherent designs directly to coherent-state identity resolution.
	\begin{proposition}
		\label{prop:identity-decomposition}
		For a set \(X=\{\n_k\}_{k=1}^{M}\subset\sphere\),
		\begin{equation}
			\frac{2J+1}{M} \sum_{k=1}^{M}P(\n_k)
			= \openone_{2J+1}
			\label{eq:sc-identity-resolution}
		\end{equation}
		holds if and only if \(X\) is an equal-weight spherical \(2J\)-design. 
		
		Equivalently, the elements 
		\begin{equation} 
			\label{eq:POVM}
			E_k=\frac{2J+1}{M} P(\n_k) 
		\end{equation} 
		form an equal-weight spin-coherent positive operator-valued measure (POVM)~\cite{Peres:2002oz} if and only if \(X\) is a spherical \(2J\)-design.
	\end{proposition}
	
	\begin{proof}
		Substituting the multipole expansion \eqref{eq:projector-expansion} into the left-hand side of \eqref{eq:sc-identity-resolution}, the \(K=0\)  contribution is exactly \(\I_{2J+1}\). Because the \(T^{(J)}_{Kq}\) are linearly independent, Eq.~\eqref{eq:sc-identity-resolution} then holds if and only if
		\begin{equation}
			\sum_{k=1}^{M}Y_{Kq}(\n_k)=0,
			\qquad
			1\le K\le 2J,
			\quad -K\le q\le K \, ,
		\end{equation}
		which is exactly the harmonic condition \eqref{eq:design-harmonic-condition} defining an equal-weight spherical \(2J\)-design.
	\end{proof}
	
	The construction extends naturally to unequal weights: for positive normalized weights \(\lambda_k>0\) with \(\sum_{k=1}^{M} \lambda_k=1\),
	\begin{equation}
		(2J+1) \sum_{k=1}^{M}  \lambda_k \, P(\n_k)=\I_{2J+1}
	\end{equation}
	holds if and only if
	\begin{equation}
		\sum_{k=1}^{M} \lambda_k Y_{Kq} (\n_k)=0,
	\end{equation}
	so weighted coherent-state identity resolutions correspond precisely to weighted spherical \(2J\)-designs.
	
	Such identity resolutions are, in general, not unique. Consider \(k\) points forming a uniformly spaced ring at polar angle \(\theta\),
	\begin{equation}
		\n_\ell(\phi_0) =
		\big(
		\sin\theta\cos(\phi_0+2\pi\ell/k),
		\sin\theta\sin(\phi_0+2\pi\ell/k),
		\cos\theta
		\big),
		\qquad
		\ell=0,\ldots,k-1 \, .
	\end{equation}
	If \(k \ge d=2J+1\), the partial sum \(\sum_{\ell=0}^{k-1}P(\n_\ell(\phi_0))\) is independent of the starting angle \(\phi_0\). Hence any coherent-state identity resolution containing such a ring may have that ring rotated rigidly about its symmetry axis without disturbing the resolution.  Examples are shown in Fig.~\ref{fig:examples-rotations}.  
	\begin{figure}[t]
		\centering
		
		\begin{subfigure}[t]{0.22\textwidth}
			\centering
			\includegraphics[width=\linewidth]{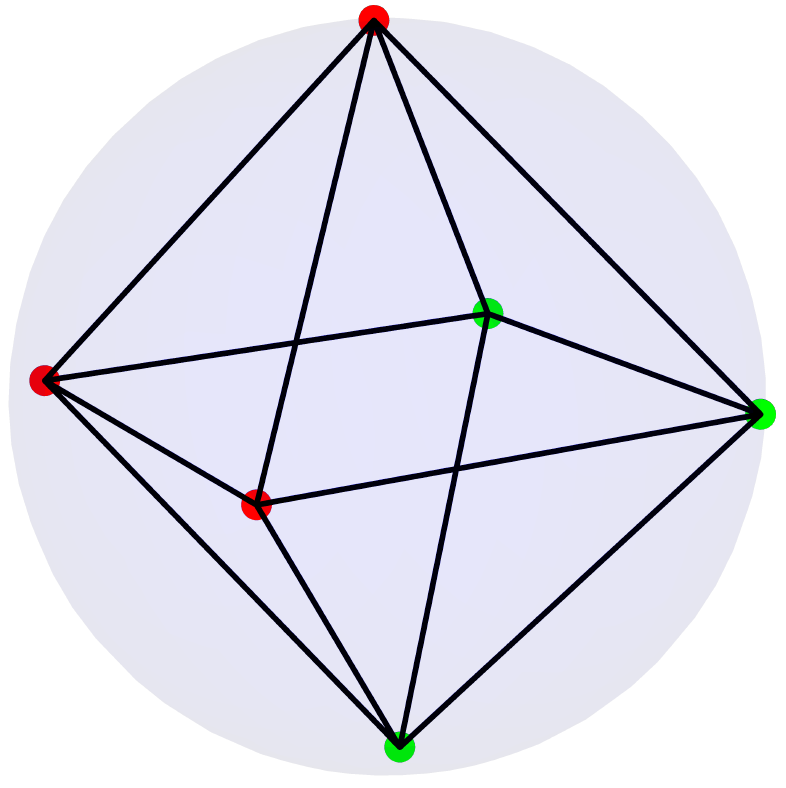}
			\caption{}
		\end{subfigure}
		\hfill
		\begin{subfigure}[t]{0.22\textwidth}
			\centering
			\includegraphics[width=\linewidth]{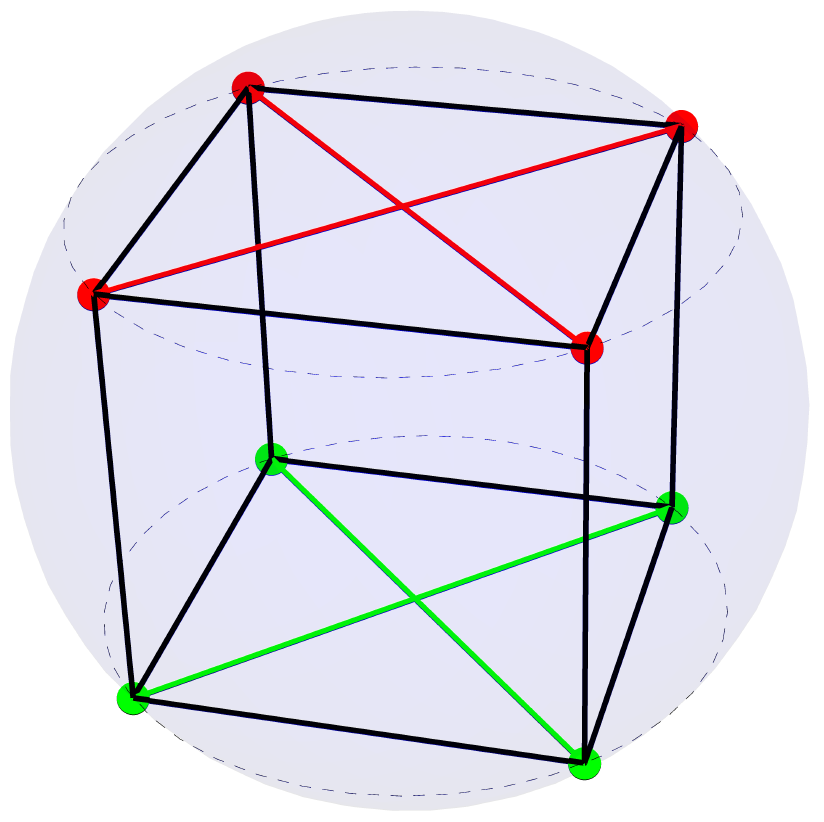}
			\caption{}
		\end{subfigure}
		\hfill
		\begin{subfigure}[t]{0.22\textwidth}
			\centering
			\includegraphics[width=\linewidth]{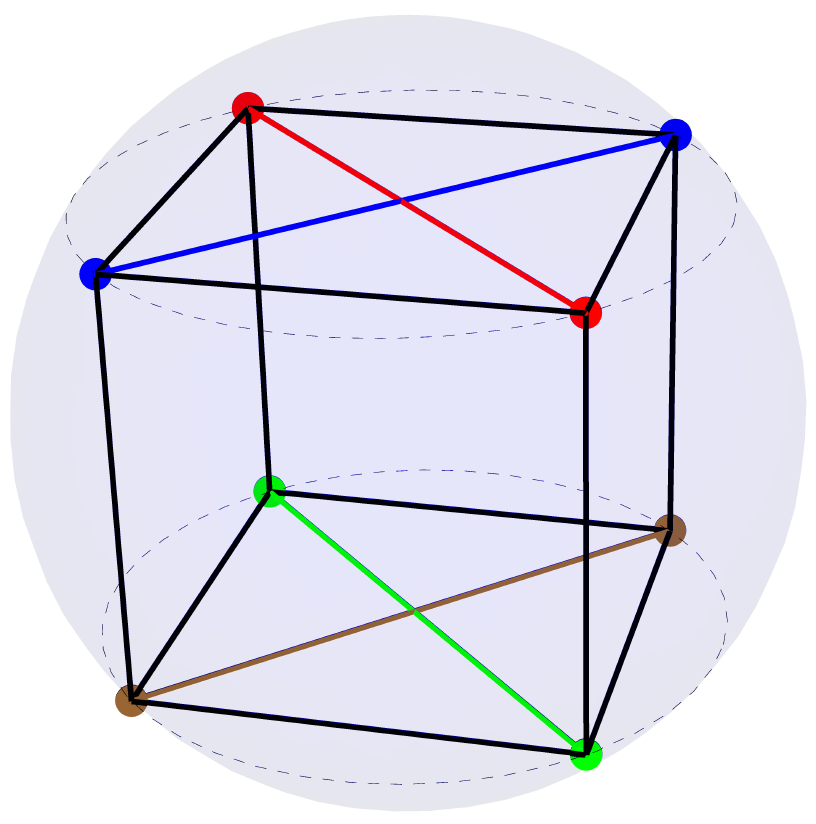}
			\caption{}
		\end{subfigure}
		\hfill
		\begin{subfigure}[t]{0.22\textwidth}
			\centering
			\includegraphics[width=\linewidth]{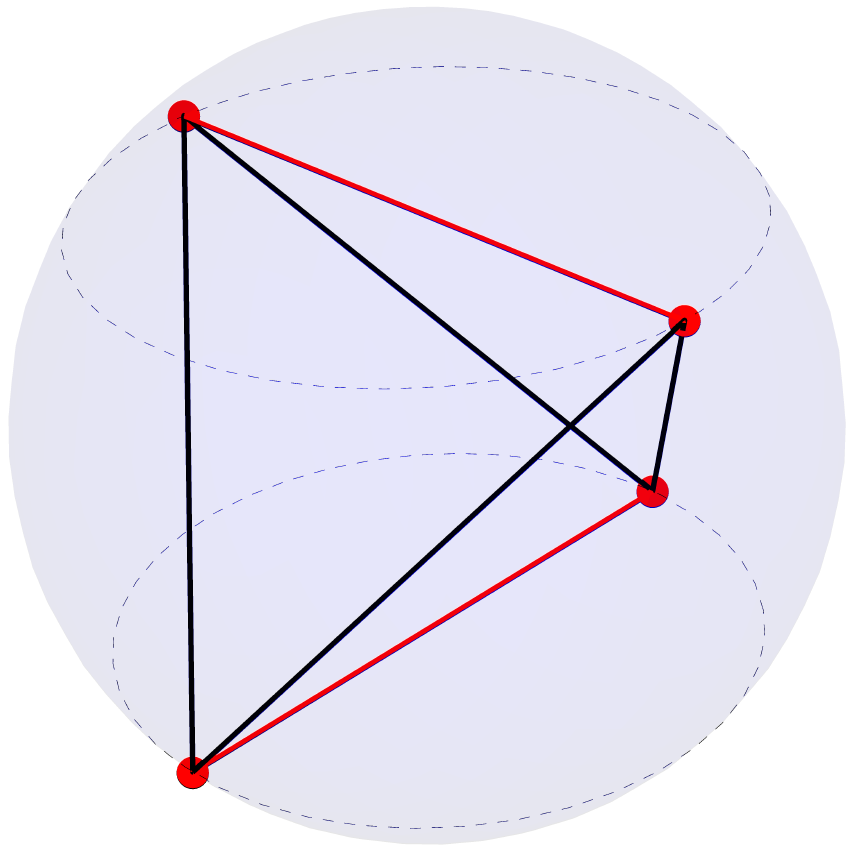}
			\caption{}
		\end{subfigure}
		
		\caption{
			Examples of azimuthal freedom in coherent-state identity resolutions.  Identity \(\I_{d}\) resolution is invariant under a collective rotation of \(k\) equally distributed points on a circle around the circle axis, provided \(k\ge d\).  The colored points indicate rigid blocks that can be rotated without changing the identity resolution in the corresponding spin dimension. a) an octahedron--identity resolution of dimensions $d=2,3$ invariant under collective rotation of points of the same color, b) a cube--invariant decomposition for $d=2,3,4$, c) a cube--invariant decomposition for $d=2$, d) a regular tetrahedron--it gives resolution of identity for $d=2,3$ and it is a specific case of c) with doubly degenerate points. }
		\label{fig:examples-rotations}
	\end{figure}

	Several familiar polyhedra furnish low-degree examples of spherical designs. The tetrahedron is a spherical \(2\)-design, the octahedron and cube are spherical \(3\)-designs, and the icosahedron and dodecahedron are spherical \(5\)-designs. By Proposition~\ref{prop:identity-decomposition}, each configuration yields an exact, equal-weight coherent-state POVM for the corresponding range of spins. The next section extends this construction beyond the identity, to successive irreducible tensor ranks.

	\section{Irreducible spin tensors and the rank theorem}
	\label{sec:rank-theorem}
	
	The decomposition \eqref{eq:operator-sector-decomposition} is standard in angular-momentum theory~\cite{Varshalovich:1988aa} and arises naturally in spin phase-space  representations~\cite{Stratonovich:1957aa,Berezin:1975mw,Varilly:1989aa,Dowling:1994sw,Brif:1999aa}.  Every operator $A \in \mathcal{B} (\Hil_J)$ admits the multipole expansion
	\begin{equation}
		A = \sum_{K=0}^{2J}\sum_{q=-K}^{K} A_{Kq} T^{(J)}_{Kq}, \qquad A_{Kq} = \Tr [ A\;  T^{(J) \dagger}_{Kq} ]  \, ,
	\end{equation}
	which organizes operator space into irreducible sectors of increasing rank. This structure is physically natural: operators of lowest rank are precisely those that appear most frequently in common Hamiltonians and  measurements.
	
	Reconstruction can therefore be studied rank by rank. For an operator supported in a single irreducible rank-$S$ sector  
	\begin{equation}
		A^{(S)}=\sum_{q=-S}^{S} A_{Sq}T^{(J)}_{Sq},
		\qquad 0\le S \le 2J \, ,
	\end{equation}
	one can define its contravariant symbol as
	\begin{equation}
		K_{A^{(S)}}(\n) = \frac{4\pi}{(2J+1) \, a_{JS}}
		\sum_{q=-S}^{S}A_{Sq} Y_{Sq}(\n) \, ,
		\label{eq:rank-r-upper-symbol}
	\end{equation}
	where the coefficients $a_{JS}$ have been defined in \eqref{eq:ajL-main}.
	In particular,
	\begin{equation}
		K_{T_{Sq}^{(J)}}(\n)=\frac{4\pi}{2J+1}\frac{Y_{Sq}(\n)}{a_{JS}}.
		\label{eq:tensor-K}
	\end{equation}
	Using Eq.~\eqref{eq:projector-expansion} together with spherical-harmonic
	orthogonality yields the continuous coherent-state representation 
	\begin{equation}
		A^{(S)} = \frac{2J+1}{4\pi}
		\int_{\sphere}K_{A^{(S)}}(\n)P(\n)\,\dd\Omega \, .
		\label{eq:rank-r-continuous-representation}
	\end{equation}
	The symbol in Eq.~\eqref{eq:rank-r-upper-symbol} has spherical degree $S$, while matrix elements of $P(\n)$ have degree at most $2J$.  This simple degree-counting argument is the
	heart of the following result, concerning the rank of spherical tensors.
	
	\begin{theorem}
		\label{thm:rank-r-quadrature}
		Let \(X=\{\n_k\}_{k=1}^{M}\subset\sphere\), and fix
		\(0\le S\le2J\). The equal-weight canonical finite coherent-state reconstruction
		formula
		\begin{equation}
			\label{eq:exp}
			A^{(S)}
			=
			\frac{2J+1}{M}
			\sum_{k=1}^{M}
			K_{A^{(S)}}(\n_k)P(\n_k)
		\end{equation}
		holds for every operator \(A^{(S)}\in\mathcal V_S\) if and only if
		\(X\) is a spherical \((2J+S)\)-design. In particular, any
		spherical \(t\)-design with \(t\ge2J+S\) reconstructs the
		entire rank-\(S\) sector exactly.
	\end{theorem}
	
	\begin{proof}
		We first prove sufficiency. For arbitrary
		\(|\psi\rangle,|\phi\rangle\in\mathcal H_J\), set
		\begin{equation}
			F_{\psi\phi}(\n)
			=
			K_{A^{(S)}}(\n)
			\langle\psi|P(\n)|\phi\rangle .
		\end{equation}
		The first factor has spherical degree \(S\) and the second has
		degree at most \(2J\), so \(F_{\psi\phi}\) has degree at most
		\(2J+S\). A spherical \(t\)-design with \(t\ge2J+S\) therefore
		gives
		\begin{equation}
			\frac1M\sum_{k=1}^{M}F_{\psi\phi}(\n_k)
			=
			\frac1{4\pi}
			\int_{\sphere}F_{\psi\phi}(\n)\,\dd\Omega .
		\end{equation}
		Multiplying by \(2J+1\) and using
		Eq.~\eqref{eq:rank-r-continuous-representation} proves reconstruction given in
		Eq.~\eqref{eq:exp}.
		
		For necessity, suppose that Eq.~\eqref{eq:exp} holds for every
		\(A^{(S)}\in\mathcal V_S\). Applying it to every basis tensor
		\(T^{(J)}_{Sq}\), substituting
		Eqs.~\eqref{eq:projector-expansion} and~\eqref{eq:tensor-K},
		and comparing the linearly independent tensors
		\(T^{(J)}_{Kq'}\) gives
		\begin{equation}
			\frac{4\pi}{M}
			\sum_{k=1}^{M}
			Y_{Sq}(\n_k)Y^*_{Kq'}(\n_k)
			=
			\delta_{SK}\delta_{qq'}
			\label{eq:sector-product-quadrature}
		\end{equation}
		for
		\[
		-S\le q\le S,\qquad
		0\le K\le2J,\qquad
		-K\le q'\le K.
		\]
		Equivalently, the discrete and continuous averages agree for
		every product \(Y_{Sq}Y^*_{Kq'}\).
		
		To extract the individual harmonic moments, use the standard
		Clebsch--Gordan product formula
		~\cite{Varshalovich:1988aa}
		\begin{equation}
			Y_{Sq}Y^*_{Kq'}
			=
			(-1)^{q'}
			\sum_{L=|S-K|}^{S+K}
			\sum_{m=-L}^{L}
			B_{SKL}\,
			C^{Lm}_{Sq,K,-q'}Y_{Lm},
			\label{eq:spherical-harmonic-product}
		\end{equation}
		where
		\begin{equation}
			B_{SKL}
			=
			\sqrt{
				\frac{(2S+1)(2K+1)}
				{4\pi(2L+1)}
			}
			C^{L0}_{S0,K0},
		\end{equation}
		and terms with \(S+K+L\) odd vanish. Clebsch--Gordan
		orthogonality gives, for every allowed \(L,m\) with
		\(B_{SKL}\ne0\),
		\begin{equation}
			Y_{Lm}
			=
			\frac1{B_{SKL}}
			\sum_{q=-S}^{S}
			\sum_{q'=-K}^{K}
			(-1)^{q'}
			(C^{Lm}_{Sq,K,-q'}\bigr)^\ast
			Y_{Sq}Y^*_{Kq'} .
			\label{eq:spherical-harmonic-product-inverse}
		\end{equation}
		Applying the difference between the discrete and continuous
		averages to
		Eq.~\eqref{eq:spherical-harmonic-product-inverse}, and using
		Eq.~\eqref{eq:sector-product-quadrature}, shows that every
		allowed \(Y_{Lm}\) is integrated exactly.
		Finally, for any \(0\le L\le2J+S\), choose
		\(K=|L-S|\).
		Then \(K\le2J\), the triangle condition is saturated, and
		\(S+K+L\) is even, so \(B_{SKL}\ne0\). Hence
		\begin{equation}
			\sum_{k=1}^{M}Y_{Lm}(\n_k)=0,
			\qquad
			1\le L\le2J+S,
			\qquad
			-L\le m\le L.
		\end{equation}
		By Eq.~\eqref{eq:design-harmonic-condition}, \(X\) is a
		spherical \(t\)-design with \(t=2J+S\). 
	\end{proof}
	
	The theorem does not assert necessity for one prescribed operator, it allows accidental exactness for special operators with point configurations below the threshold. 
	Its strength lies in the equivalence between single geometric condition on the design and exact reconstruction of the entire irreducible rank-$S$ sector.
	
	The condition $t \ge 2J+S$ also has a transparent operational reading.  For a spherical $t$-design with $t\ge 2J$, 
	\red{the highest rank guaranteed to be reconstructed exactly is}
	\begin{equation}
		R_{\max} = \min\{2J,\,t-2J\}.
	\end{equation}
	Equivalently, for any integer $0\le R\le 2J$, designs satisfying $t\ge 2J+R$ reconstruct all irreducible sectors $0\le S \le R$ simultaneously, and $t\ge 4J$ suffices for reconstruction of the full operator space $\mathcal{B}(\Hil_{J})$.
	
	More generally, if $A = \sum_{S=0}^{R} A^{(S)}$, with $A^{(S)}\in V_{S}$, then $t\ge 2J+R$ and linearity of the theorem together imply
	\begin{equation}
		A = \frac{2J+1}{M} \sum_{k=1}^{M} K_{A}(\n_{k})\,P(\n_{k}),
	\end{equation}
	where $K_{A}(\n)=\sum_{S=0}^{R}K_{A^{(S)}}(\n)$.
	This includes all symmetrized polynomials in the spin operators whose tensor rank is at most $R$.
	
	The practical payoff is immediate.  The two most common measurement tasks, estimating expectation values of spin operators and their variances, correspond to ranks $S=1$ and $S=2$, respectively. Theorem~\ref{thm:rank-r-quadrature} therefore guarantees that sampling from designs of order $2J+1$ and $2J+2$ suffices for these tasks exactly, without any recourse to full quantum state tomography.
	
	\begin{figure}[t]
		\centering
		\includegraphics[width=\linewidth]{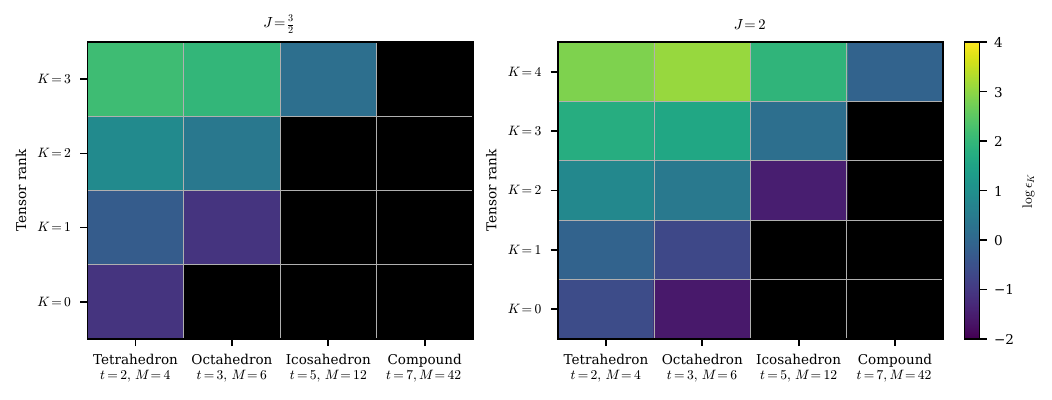}
		\caption{ Hilbert--Schmidt reconstruction error for irreducible spin-tensor operators using a coherent-state decomposition generated by analytical spherical designs given by chosen Platonic solids~\cite{GoethalsSeidel1979} and a compound of seven octahedra~\cite{Rudzinski2026}. For each spin \(J\), tensor rank \(K\), and design, the maximal Hilbert--Schmidt distance between spin-tensor \(T_{Kq}^{(J)}\) and its reconstruction via coherent design is plotted.  Black cells denote analytically exact reconstruction. 	}
		\label{fig:tensor-error-heatmaps}
	\end{figure}
	
	In general, when the design order falls below the threshold, the discrete sum no longer reproduces the target operator exactly.  The resulting error has a precise algebraic expression in terms of the spherical-harmonic Gram matrix of the
	point set.  For $X=\{\n_{k}\}_{k=1}^{M}$, define the reconstruction map
	\begin{equation}
		\mathcal{R}_{X}(A) =
		\frac{2J+1}{M}
		\sum_{k=1}^{M} K_{A}(\n_{k})\,P(\n_{k}).
	\end{equation}
	For a tensor basis element $T_{Kq}^{(J)}$, a direct computation gives
	\begin{equation}
		\mathcal{R}_{X}(T_{Kq}^{(J)}) = \sum_{K'=0}^{2J}\sum_{q'=-K'}^{K'}
		\frac{a_{JK'}}{a_{JK}}\, G^{X}_{Kq,K'q'}\, T_{K'q'},
	\end{equation}
	where
	\begin{equation}
		G^{X}_{Kq,K'q'} = \frac{4\pi}{M} \sum_{k=1}^{M} Y_{Kq}(\n_{k})\,Y^{*}_{K'q'}(\n_{k})
	\end{equation}
	is the discrete harmonic Gram matrix.  The Hilbert--Schmidt reconstruction
	error is therefore
	\begin{equation}
		\bigl\|T_{Kq}^{(J)}- \mathcal{R}_{X}(T_{Kq}^{(J)}) \bigr\|_{\mathrm{HS}}^{2}
		= \sum_{K'=0}^{2J}\sum_{q'=-K'}^{K'}	\left| 	\delta_{KK'}\delta_{qq'}	-
		\frac{a_{JK'}}{a_{JK}}\,G^{X}_{Kq,K'q'}	\right|^{2} \, . 
	\end{equation}
	For a spherical $t$-design, the off-diagonal entries of $G^{X}$ vanish whenever $K+K'\le t$, so that
	$G^{X}_{Kq,K'q'}=\delta_{KK'}\delta_{qq'}$ in that range, and the error is exactly zero, in precise agreement with
	Theorem~\ref{thm:rank-r-quadrature}.
	
	To visualize how the error decays as the design order increases, in Fig.~\ref{fig:tensor-error-heatmaps} we plot the rank-$K$ componentwise maximal error
	\begin{equation}
		\epsilon_{K}(X,J) = \max_{-K\le q \le K} 	\bigl\|T^{(J)}_{Kq}-\mathcal{R}_{X}(T_{Kq}^{(J)})\bigr\|_{\mathrm{HS}}
	\end{equation} for a range of spins $J$, ranks $K$, and designs generated by Platonic solids and related configurations. The black cells--exact reconstruction--trace out precisely the region $t\ge 2J+K$ predicted by the theorem, providing a vivid confirmation of its sharpness.

		\begin{figure*}[t]
			\centering
			\includegraphics[width=\textwidth]{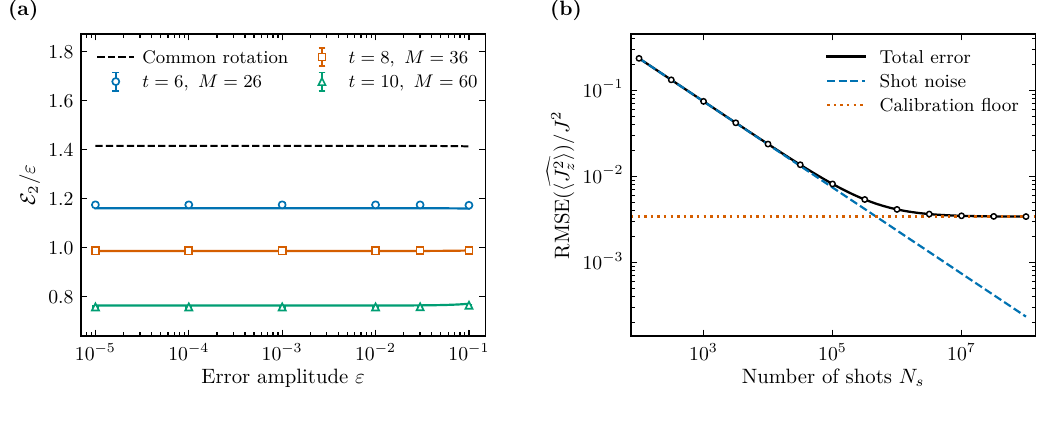}
			\caption{\red{Direction errors in spin-$2$ reconstruction.
				(a) The root-mean-square (RMS) Hilbert--Schmidt error $\mathcal E_S$ from Eq.~\eqref{eq:sector-rms} applied to the rank-$2$ sector, divided by the RMS angular disturbance $\varepsilon$. 
				 For each design of degree $t$ with $M$ points, the symbols illustrate Monte Carlo simulation results with error bars invisible in the scale of the picture, while solid lines give the model predictions. The error caused by a common rigid rotation by an angle $\varepsilon$ serves as a reference point. 	 (b) Root-mean-square error of the estimate of
				$\langle J_z^2\rangle$, divided by the operator norm $J^2=4$, for a coherent state along $(1,2,3)/\sqrt{14}$, with $t=6$, $M=26$, and $\varepsilon=10^{-2}$. Points show the total error. The solid line represents its analytical prediction given in Eq.~\eqref{eq:RMSE} as a combination of the shot noise (dashed line) and RMS calibration bias (dotted line). More shots reduce statistical uncertainty but leave the calibration floor unchanged.
                }
                }
			\label{fig:robustness}
		\end{figure*}
	
	\red{
		\subsection{Sensitivity to imperfect measurement directions}
		\label{sec:robustness}

        Take a spherical $t$-design with $t\ge 2J+S$ and an operator $A$ supported on ranks up to $S$, so that $\mathcal R_X(A)=A$.         
        Suppose that the projections are intended along the directions $\n_k$ but the measurement actually implements $\widetilde{\n}_k$, which differs from $\n_k$ by an angle $\delta_k=\arccos(\n_k\cdot\widetilde{\n}_k)$.
        Keeping the coefficients $K_A(\n_k)$ fixed gives the perturbed reconstruction map
		\begin{equation}
			\widetilde{\mathcal R}_X(A)
			=\frac{2J+1}{M}\sum_k K_A(\n_k)P(\widetilde\n_k).
			\label{eq:noisy-map}
		\end{equation}
		Independently perturbed effects $(2J+1)P(\widetilde\n_k)/M$ need not sum
		to the identity.  The coherent-state overlap gives
		$\|P(\widetilde\n_k)-P(\n_k)\|_{\HS}
		=\sqrt{2[1-\cos^{4J}(\delta_k/2)]}\le\sqrt J\,\delta_k$,and, by the triangle inequality, we obtain the following bound
		\begin{equation}
			\|\widetilde{\mathcal R}_X(A)-A\|_{\HS}
			\le\frac{(2J+1)\sqrt J}{M}\sum_k |K_A(\n_k)|\delta_k.
			\label{eq:robust-bound}
		\end{equation}
		To model independent setting errors, draw a Gaussian displacement in the tangent plane independently at each nominal direction. Its two orthogonal components are independent, each with mean zero and variance $\varepsilon^2/2$. Each $\widetilde\n_k$ is given by moving $\n_k$ along the sphere in the direction of the drawn vector by an angle equal to its length. Thus $\varepsilon$ is the root-mean-square (RMS) angular step.        Averaging over calibrations, denoted by $\E$, gives the mean squared Hilbert--Schmidt error for a normalized rank-$S$ tensor $T_{Sq}$,
		\begin{equation}
			\E\left[\|\widetilde{\mathcal R}_X(T_{Sq})-T_{Sq}\|_{\HS}^2\right]
			=\frac{4\pi J}{M a_{JS}^2}\varepsilon^2+O(\varepsilon^4).
			\label{eq:robust-rms}
		\end{equation}
		Independence eliminates cross terms between the centered projector
		fluctuations. Expanding the squared projector distance as
		$J\delta_k^2+O(\delta_k^4)$ and using the design identity
		$M^{-1}\sum_k|K_{T_{Sq}}(\n_k)|^2=4\pi/(d^2a_{JS}^2)$
		then gives the leading term in Eq.~\eqref{eq:robust-rms}.
		
		Moreover the mean reconstructed tensor is
		\begin{equation}
			\E\left[\widetilde{\mathcal R}_X(T_{Sq})\right]
			=[1-S(S+1)\varepsilon^2/4+O(\varepsilon^4)]T_{Sq}.
		\end{equation}
        Averaging over isotropic perturbations multiplies each rank-$S$ harmonic by $\E[P_S(\cos\delta_k)]$, where $P_S$ is the Legendre polynomial of degree $S$. Its small-angle expansion gives the formula.
		
		Fig.~\ref{fig:robustness}(a) compares three spherical designs~\cite{SloaneRepository} for spin $J=2$. 	The RMS error formula for the sector of tensor rank-$S$ is given by
		\begin{equation}
			\mathcal E_S=\left[(2S+1)^{-1}\sum_{q=-S}^{S}\mathbb E\left[\|\widetilde{\mathcal R}_X(T_{Sq})-T_{Sq}\|_{\HS}^2\right]\right]^{1/2}.
            \label{eq:sector-rms}
		\end{equation}
        For $J=S=2$, it gives $\mathcal E_2\simeq\varepsilon\sqrt{35/M}$ for small independent disturbances. The results are compared with a rigid rotation of the entire set, whose error is the same for all three designs. 
		
        Panel (b) shows the measurement of the expectation value $\mu=\langle J_z^2\rangle=41/14$ for the coherent state in the direction $\n=(1,2,3)/\sqrt{14}$. In each calibration $r$,
the $M=26$ measurement settings, given by points of the spherical design, are disturbed independently and held fixed for all $N_s$ shots. Each shot selects a setting $k$ uniformly and tests whether the spin projection along the implemented direction $\widetilde{\n}^{(r)}_k$ is $+2$. 
This outcome has Born probability $q_{rk}=[(1+\n_0\cdot\widetilde{\n}^{(r)}_k)/2]^4$ and receives the nominal score $s_k=5K_{J_z^2}(\n_k)=\tfrac52(21n_{k,z}^2-3)$, where $n_{k,z}$ is the intended direction's $z$ component. All other outcomes receive a score of zero. The single-shot mean and variance are $m_r=M^{-1}\sum_k q_{rk}s_k$ and $v_r=M^{-1}\sum_k q_{rk}s_k^2-m_r^2$. The mean $\widehat\mu$ of $N_s$ independent scores has bias $b_r=m_r-\mu$ and variance $v_r/N_s$. Its normalized RMS error is
\begin{equation}
    \frac{\mathrm{RMSE}(\widehat\mu)}{J^2}
  =\frac14\sqrt{\overline{b_r^2}+\frac{\overline{v}_r}{N_s}},
  \label{eq:RMSE}
\end{equation}
where the overbar denotes the average over calibrations and $J^2=4$.
	}

	\section{Weighted Gauss--Legendre construction}
	\label{sec:gauss-legendre}
	
	The spherical designs used in the preceding section have the elegant property of equal weights, but their explicit construction can be nontrivial.  Here we present a complementary approach that is fully explicit and constructive at the cost of introducing nonuniform weights: a Gauss--Legendre quadrature in the polar variable combined with a uniform grid in the azimuthal variable. For
	identity resolution, this construction 
	appeared in
	the literature on finite optimal spin measurements~\cite{BaganBaigMunozTapia2001,IblisdirRoland2006}.
	We recall it here and extend
	its use to the uniform reconstruction of arbitrary tensor rank
	\(S\) by increasing the required polynomial degree from \(2J\)
	to \(2J+S\).
	
	Let $R_{\mathrm{max}}$ be the maximal spherical polynomial degree to be integrated exactly. Choose $R$ Gauss--Legendre nodes $x_\alpha \in (-1,1)$ ($\alpha = 1,\ldots,R$); i.e.,  zeros of the Legendre polynomial $P_R(x)$, with associated weights
	\begin{equation}
		g_\alpha = \frac{2}{(1-x_\alpha^2)\bigl[P_R'(x_\alpha)\bigr]^2} \, .
	\end{equation}
	These weights satisfy
	\begin{equation}
		\sum_{\alpha=1}^{R}g_\alpha f(x_\alpha)  =	\int_{-1}^{1} f(x)\,\dd x
	\end{equation}
	for every polynomial \(f\) of degree at most \(2R-1\).  Exact integration up to degree $R_{\mathrm{max}}$ therefore requires 	$2R-1\ge R_{\mathrm{max}}$.  	
	
	In the azimuthal direction, choose $Q$ equally spaced angles $\phi_\beta= 2\pi\beta/Q$, with  $\beta=0,\ldots,Q-1$ and $Q>R_{\mathrm{max}}$. The resulting grid points on the sphere are
	\begin{equation}
		\n_{\alpha\beta} =
		\left(
		\sqrt{1-x_\alpha^2}\cos\phi_\beta,\,
		\sqrt{1-x_\alpha^2}\sin\phi_\beta,\,
		x_\alpha
		\right),
	\end{equation}
	with normalized weights $w_{\alpha\beta} = g_\alpha /(2Q)$.  We now verify that this weighted grid integrates all spherical polynomials of degree at most $R_{\mathrm{max}}$ exactly.  It suffices to check spherical harmonics. Since $Y_{Kq}(\n_{\alpha\beta}) = 	N_{Kq} P_K^{|q|}(x_\alpha)e^{i q \phi_\beta}$, where \(N_{Kq}\) is a normalization constant,  two cases arise. If  $q \ne 0$, the azimuthal sum has to vanish: $\sum_{\beta=0}^{Q-1}e^{iq\phi_\beta}=0$ because \(|q|\le K\le R_{\mathrm{max}}<Q\).  If \(q=0\), the polar sum vanishes by Gauss--Legendre exactness and the orthogonality of Legendre polynomials:
	\begin{equation}
		\sum_{\alpha=1}^{R} g_\alpha P_K (x_\alpha)
		= 	\int_{-1}^{1}P_K(x)\,\dd x 	= 	0,
		\qquad
		1\le K\le R_{\mathrm{max}} \, .
	\end{equation}
	Combining both cases with the normalization,  one obtains
	\begin{equation}
		\sum_{\alpha=1}^{R} \sum_{\beta=0}^{Q-1}
		w_{\alpha\beta}\, f(\n_{\alpha\beta})
		=
		\frac{1}{4\pi}\int_{\sphere} f(\n)\,\dd \Omega
	\end{equation}
	for every spherical polynomial $f$ of degree at most $R_{\mathrm{max}}$.
	
	For the coherent-state resolution of the identity in dimension $2J+1$, one needs $R_{\mathrm{max}}= 2J$.  
	The smallest choice within this separated-grid construction satisfying
	\(Q>R_{\max}\) is
	$ R= \lceil J+1/2 \rceil$ and $ Q=2J+1$, which gives an explicit weighted coherent-state POVM
	\begin{equation}
		E_{\alpha\beta} = (2J+1)w_{\alpha\beta} P(\n_{\alpha\beta}),
	\end{equation}
	or, equivalently,
	\begin{equation}
		\I_{2J+1} = (2J+1) \sum_{\alpha=1}^{R}\sum_{\beta=0}^{Q-1} 	w_{\alpha\beta}P(\n_{\alpha\beta}).
	\end{equation}
	
	The Gauss--Legendre grid therefore furnishes a positive, tight, coherent-state frame.  Unlike an equal-weight spherical design, its weights vary with latitude through the factors $g_\alpha$.  The total number of coherent states required for identity resolution is
	\begin{equation}
		M_{\mathrm{GL}} = RQ = (2J+1)\left\lceil\frac{2J+1}{2}\right\rceil.
	\end{equation}
	In terms of the Hilbert-space dimension $d = 2J+1$, this scales as $ M_{\mathrm{GL}}  \sim  d^2/2,$  so that $d \sim \sqrt{2M_{\mathrm{GL}}}$ is the largest accessible dimension for a given number of measurement points.
	
	The same technique  can be used for rank-\(S\) tensor reconstruction.	In that case the required degree is \(R_{\mathrm{max}}=2J+S\).	Indeed, the proof of Theorem~\ref{thm:rank-r-quadrature} only uses exact integration of the product \(K_{A^{(S)}}(\n)\,\bra{\psi}P(\n)\ket{\phi}\), which has spherical degree at most \(R_{\mathrm{max}}=2J + S\).  
	Therefore, if \(2R-1\ge 2J+S\), and \(Q>2J+S\),  
	the Gauss--Legendre grid gives the following positive-weight counterpart of the equal-weight reconstruction formula~\eqref{eq:exp}:
	\begin{equation}
		A^{(S)}
		=
		(2J+1)
		\sum_{\alpha=1}^{R}\sum_{\beta=0}^{Q-1}
		w_{\alpha\beta}
		K_{A^{(S)}}(\n_{\alpha\beta}) 	P(\n_{\alpha\beta}).
	\end{equation}
	Within this 
	construction, a minimal choice satisfying the two constraints is $ R = \lceil \tfrac12 (2J + S + 1) \rceil$ and $Q = 2J + S + 1$. The corresponding number of coherent states is
	\begin{equation}
		M_{\rm GL}^{(S)} = (2J+S+1)\left\lceil\tfrac{1}{2}(2J+S+1)\right\rceil .
		\label{eq:GL-count-rank-r}
	\end{equation}
	
	More generally, wherever the equal-weight design average appears in a reconstruction identity whose proof rests solely on exact polynomial  integration, it can be replaced by the weighted quadrature average via the substitution
	\begin{equation}
		\frac{1}{M}\sum_{k=1}^{M}
		\quad\longrightarrow\quad
		\sum_{\alpha=1}^{R}\sum_{\beta=0}^{Q-1} w_{\alpha\beta}.
	\end{equation}
	The Gauss--Legendre construction thus provides a systematic, explicit
	alternative to spherical designs for any reconstruction task that depends
	only on exact integration up to a prescribed degree.

	\section{Low-rank kernels and finite-shot estimators}
	\label{sec:low-rank-estimators}
	
	The rank theorem gives exact operator decompositions.  We now spell out its most direct operational consequence: low-rank spin observables can be estimated from the outcomes of a spin-coherent design measurement by fixed classical post-processing. No
	observable-specific tomography is needed.
	
	Throughout, we sample directly from the equal-weight coherent-state POVM \eqref{eq:POVM}, and write  \(p_k=\Tr(\varrho E_k)\)  for the outcome probabilities of a state \(\varrho\). If an operator admits the finite coherent-state expansion \eqref{eq:exp}, then  its expectation value is a simple linear function of the same outcome probabilities:
	\begin{equation}
		\langle A\rangle_\varrho = \sum_{k=1}^M p_k\,K_A(\n_k).
	\end{equation}
	Here \(K_A\) is a fixed classical post-processing function —a ``score'' assigned to each
	outcome— and should not be confused with an eigenvalue of \(A\).
	
	The spin vector \(J_x,J_y,J_z\) sits in the rank-one sector, with kernel $ K_{J_a}(\n)=(J+1) \, n_a$ $(a \in \{x,y,z\})$.
	Thus, for a design with \(t\ge 2J+1\),
	\begin{equation}
		\mu_a \equiv \langle J_a\rangle_\varrho =
		(J+1)\sum_{k=1}^{M} p_k n_{k,a}.
		\label{eq:mean-spin-estimator-expectation}
	\end{equation}
	
	The symmetrized second moment
	\(C_{ab}=\tfrac12\{J_a,J_b\}\) decomposes into a scalar part and a
	traceless rank-two part. For \(J\ge1\), its nontrivial anisotropic component
	therefore lies in the rank-two sector.
	Its kernel is 
	$K_{C_{ab}}(\n) = \tfrac{1}{2}[(J+1)(2J+3)] n_an_b - \tfrac{1}{2} (J+1) \delta_{ab}$  
	and now, for \(t\ge 2J+2\), 
	\begin{equation}
		\langle C_{ab}\rangle_\varrho
		= \sum_{k=1}^M p_k\left[\frac{(J+1)(2J+3)}{2}n_{k,a}n_{k,b} - \frac{J+1}{2}\delta_{ab}\right],
		\label{eq:second-moment-from-design-probabilities}
	\end{equation}
	and the spin covariance matrix is 
	\begin{equation}
		\Gamma_{ab} = \langle C_{ab}\rangle_\varrho - \mu_a\mu_b.
		\label{eq:spin-covariance-matrix}
	\end{equation}
	In short: once the design degree reaches \(t\ge 2J+2\), the same measurement data fix every first and second spin moment.
	
	In practice, the \(p_k\) are never known exactly; only \(N_s\) independent outcomes \(k_1,\dots,k_{N_s}\) are sampled from \eqref{eq:POVM}. Each outcome contributes a score, obtained by evaluating the relevant kernel at the sampled
	direction:
	\begin{equation}
		\begin{aligned}
			X_a(k) &= (J+1)n_{k,a}, \\
			Y_{ab}(k) &= \frac{(J+1)(2J+3)}{2}n_{k,a}n_{k,b} - \frac{J+1}{2}\delta_{ab}.
		\end{aligned}
	\end{equation}
	Averaging these scores gives unbiased estimators of the mean spin and the second moment, for any \(N_{s}\):
	\begin{equation}
		\hat\mu_a = \frac{1}{N_{s}} \sum_{i=1}^{N_{s}} X_a(k_i), \qquad
		\hat C_{ab} = \frac{1}{N_{s}}\sum_{i=1}^{N_{s}} Y_{ab}(k_i).
	\end{equation}
	Since the outcomes are independent and
	\(\lvert X_a(k)\rvert\le J+1\), the variance of the mean-spin
	estimator obeys the simple state-independent bound
	\begin{equation}
		\operatorname{Var}(\widehat{\mu}_a)
		=
		\frac{\operatorname{Var}(X_a)}{N_s}
		\le
		\frac{(J+1)^2}{N_s}.
		\label{eq:var-bound}
	\end{equation}
	Thus its standard deviation decreases at least as
	\(N_s^{-1/2}\).
	
	Note that  \(\Gamma_{ab}\) is a nonlinear function of \(\mu_a\), so the naive plug-in \(\hat C_{ab}-\hat\mu_a\hat\mu_b\) is biased at finite \(N_{s}\): the product \(\hat\mu_a\hat\mu_b\) mixes same-shot pairs (\(i=l\)) with different-shot pairs (\(i\ne l\)), and only the latter average to \(\mu_a\mu_b\). This is the same issue that Bessel's correction fixes for the ordinary sample variance, here applied to a cross term. Writing
	\begin{equation}
		\hat q_{ab} = \frac{1}{N_{s}}\sum_{i=1}^{N_{s}} X_a(k_i)X_b(k_i)
	\end{equation}
	for the diagonal (same-shot) piece, the corrected, unbiased estimator for \(N_{s}>1\) is
	\begin{equation}
		\hat\Gamma_{ab} = \hat C_{ab} - \frac{N_{s}\hat\mu_a\hat\mu_b - \hat q_{ab}}{N_{s}-1}.
	\end{equation}
	
	\begin{proposition}
		If \(t\ge 2J+2\) and \(N_{s }>1\), then \(\mathbb E[\hat\Gamma_{ab}]=\Gamma_{ab}\), where $\mathbb E [\cdot ]$ is the average with respect to the distribution $p_k$.
	\end{proposition}
	\begin{proof}
		Since \(\hat C_{ab}\) is a sample mean of \(Y_{ab}\), \(\mathbb E[\hat C_{ab}]=\langle C_{ab}\rangle_\varrho\).
		For the correction term,
		\begin{equation}
			\frac{N_{s}\hat\mu_a\hat\mu_b - \hat q_{ab}}{N_{s}-1} = \frac{1}{N_{s}(N_{s}-1)}\sum_{i\ne l} X_a(k_i)X_b(k_l).
		\end{equation}
		Each \(i\ne l\) term pairs two independent shots, so has expectation \(\mu_a\mu_b\); averaging over all \(N_{s}(N_{s}-1)\) such pairs leaves the expectation unchanged. Hence
		\begin{equation}
			\mathbb E[\hat\Gamma_{ab}] = \langle C_{ab}\rangle_\varrho - \mu_a\mu_b = \Gamma_{ab}. \qedhere
		\end{equation}
	\end{proof}

	The design condition \(t\ge 2J+2\) guarantees these estimators are unbiased for spin observables up to rank two, for any state \(\varrho\). 
	The bound~\eqref{eq:var-bound} controls the finite-shot fluctuations of
	the mean-spin estimator and displays the expected \(N_s^{-1/2}\)
	scaling of its statistical uncertainty.

	\red{
		\subsection{Worked example: spin moments from twelve outcomes}
		\label{sec:worked-example}
		
		Take \(J=1\) and the twelve vertices \(\n_k\) of the regular
		icosahedron, which are given in Appendix~\ref{app:explicit-decompositions}. Since this is
		a spherical \(5\)-design, it exceeds the requirement \(2J+2=4\) for
		both first and second moments. The measurement effects are
		\(E_k=P^{(1)}(\n_k)/4\). To illustrate every step, consider a state
		\(\varrho=(|1,1\rangle\langle1,1|+|1,0\rangle\langle1,0|)/2\)
		and write \(z_k=n_{k,z}\). The coherent-state overlaps
		\(|\langle\n_k|1,1\rangle|^2=(1+z_k)^2/4\) and
		\(|\langle\n_k|1,0\rangle|^2=(1-z_k^2)/2\) give
		\begin{equation}
			p_k=\Tr(\varrho E_k)=\frac{3+2z_k-z_k^2}{32}.
			\label{eq:worked-probabilities}
		\end{equation}
		The vertices satisfy
		\(\sum_k z_k=\sum_k z_k^3=0\),
		\(\sum_k z_k^2=4\), and \(\sum_k z_k^4=12/5\).
		Consequently, \(\sum_kp_k=1\),
		\(\sum_kp_kz_k=1/4\), and \(\sum_kp_kz_k^2=3/10\).
		Evaluating the relevant kernels at $J=1$ gives \(X_z(k)=2z_k\) and
		\(Y_{zz}(k)=5z_k^2-1\). Eqs.~\eqref{eq:mean-spin-estimator-expectation}
		and~\eqref{eq:second-moment-from-design-probabilities} yield
		\begin{equation}
			\langle J_z\rangle=\frac12,\qquad
			\langle J_z^2\rangle=\frac12,\qquad
			\Var(J_z)=\frac14.
			\label{eq:worked-moments}
		\end{equation}
		These values agree with direct evaluation in \(\varrho\). In an
		experiment, the \(p_k\) are replaced by the observed frequencies.
		The sample means of these two scores estimate the first and second
		moments, and the correction above gives an unbiased estimate of the
		variance. Here the mean-spin estimator has
		\(\Var(\hat\mu_z)=19/(20N_s)\), since
		\(\sum_kp_kX_z(k)^2=6/5\).
		
		\section{Extension to other compact groups}
		\label{sec:other-groups}
		
		The mechanism extends beyond SU(2), although its scalar degree threshold
		does not. Let $U$ be a $d$-dimensional irreducible unitary representation
		of a compact group $G$, and $P(g)=U(g)P_0U(g)^\dagger$ a rank-one orbit, where $P_0$ is a trace-one rank-one projector.
		With the normalized Haar measure, Schur's lemma gives $d\int_G P(g)\,\dd g=\I_d$.
		If the orbit spans $\B(\C^d)$, its frame map, defined as
		\begin{equation}
			\mathcal F(A)=d\int_G\Tr[A P(g)]P(g)\,\dd g,
			\label{eq:group-frame}
		\end{equation}
		is invertible. Then $K_A(g)=\Tr[\mathcal F^{-1}(A)P(g)]$ gives
		$A=d\int_G K_A(g)P(g)\,\dd g$. 
        A positive cubature is a finite weighted sum with positive weights that reproduces the Haar integral on specified functions. Exactness on the matrix elements of $K_A(g)P(g)$ gives a finite reconstruction for each target $A$~\cite{Brif:1999aa,Dariano2004}. Including the identity among the targets ensures that the effects $d\,w_kP(g_k)$ sum to the identity. The decomposition of $U\otimes\overline U$ determines the blocks of the frame map; for repeated components, these blocks act on multiplicity spaces.
		
		For example, in the defining representation of SU($d$), the orbit is
		$\C\mathbf P^{d-1}$. Any normalized weighted complex-projective
		$2$-design $\{w_k,P_k\}$ gives~\cite{Scott2006}
		\begin{equation}
			A=d\sum_k w_k\big[(d+1)\Tr(A P_k)-\Tr A\big]P_k.
			\label{eq:sud-example}
		\end{equation}
        A weighted projective 2-design reproduces the uniform average of $P\otimes P$ over pure-state projectors.
	       Contracting the projective second-moment identity with $A$ gives
		$\mathcal F(A)=[A+(\Tr A)\I]/(d+1)$, while its inverse reads
		$\mathcal F^{-1}(A)=(d+1)A-(\Tr A)\I$ leading to Eq.~\eqref{eq:sud-example}.
		
		Thus the finite decomposition extends to other groups, while the particular condition
		$t\ge2J+S$ reflects SU(2) angular-momentum addition. Noncompact groups
		require additional integrability and cubature analysis.
	}

	\section{Concluding remarks}
	\label{sec:conclusions}
	
	An overcomplete family of states resolves the identity  through a continuum of outcomes. For canonical coherent states this is not a practical obstacle: the von Neumann lattice replaces the continuous integral with a discrete, evenly spaced grid, exactly at the minimal point density a shift-invariant structure allows. 
	For spin, finite coherent-state tomography and discrete operator bases were
	already known. The question addressed here is when the canonical
	contravariant-symbol integral itself is reproduced by a finite equal-weight sum
	uniformly on an irreducible tensor sector.
	We have shown that spherical designs answer this more specific question:
	spherical designs give a discrete substitute for the coherent-state continuum on \(\sphere\), in the precise sense that exact reconstruction of the rank-\(S\) sector of operator space on \(\mathcal H_J\) is equivalent to the defining averaging property of a spherical \(t\)-design with \(t\ge 2J+S\).
	
	The number of coherent states needed for any fixed-rank sector grows only quadratically in the design degree, and therefore only polynomially in the Hilbert-space dimension, rather than exponentially. Unlike the planar lattice, whose unit cell has fixed area \(h\) independent of the operator being probed, the spherical construction has no translation symmetry, and the required design order depends on both the spin \(J\) and the target rank \(S\): the cell area decreases as \(J^{-2}\) rather than staying fixed. This difference is why the argument here is a degree-counting argument rather than a lattice construction.
	
	The two constructions we give serve different purposes. 
	\red{Spherical designs give an exact reconstruction with equal weights, although 
    explicit designs of the required degree can be difficult to find.}
	The weighted Gauss--Legendre grid is an explicit alternative for any spin and rank, at the cost of abandoning equal weights. In this explicit protocol, identity resolution requires \(d \lceil d/2\rceil\) coherent states, whereas reconstruction of the full operator space, and hence full quantum-state tomography, requires \(d(2d-1)\) coherent states, where \(d=2J+1\).
	Together they guarantee that a usable POVM is always available. In both cases the POVM is built from coherent states, which are straightforward to prepare and measure; this is the difference between a spin-coherent design and a generic tensor design, whose defining average runs over the whole of Hilbert space and need not consist of physically simple states.
	
	The practical consequence is that reconstruction reduces to fixed classical post-processing: fix the POVM once, record outcomes, and recover any low-order spin moment  by evaluating a fixed kernel and averaging, with no adaptive measurement and no full state tomography. 
	
	Some questions remain open. 
    \red{Unequal weights can reduce the number of directions needed for identity resolution~\cite{Latorre1998}. Determining the minimum number of directions for general spin and tensor rank remains an open problem.}
    More broadly, spin-coherent designs join a small set of finite structures that replace continuous averages, including weighted complex-projective designs used for optimal state discrimination~\cite{Goyenecheetal2015,Goyenecheetal2018}
	and the discrete sums appearing in lattice gauge theory~\cite{Ruhl1982}.  The plane and the sphere both do; what determines the answer for other phase spaces is left for future work.  
	\red{Section~\ref{sec:other-groups} identifies the orbit-frame and cubature conditions for compact groups. However, sharp degree bounds and minimal constructions on other phase spaces remain open problems.}

	\appendix

	\section{Contravariant symbols} 
	\label{app:tensor-normalizations}
	
	This appendix fixes the angular-momentum conventions used in the main text and gives the normalization constants needed to reproduce the numerical calculations.   Recall the spin-coherent state pointing along \(\n\), $	|\n\rangle_J =	U_J(R_{\n})|J,J\rangle$, 	where \(R_{\n}\) sends the north pole to	\(\n\).  The coherent-state projector at the north pole reads
	\begin{equation}
		P^{(J)}(\bm e_z)=|JJ\rangle\langle JJ|.
	\end{equation}
	Using the definition of the irreducible tensors, only the term \(q=0\) contributes, and
	\begin{equation}
		P^{(J)}(\bm e_z)
		=
		\sum_{K=0}^{2J}
		b_{JK}\,T^{(J)}_{K0},
		\qquad
		b_{JK}
		=
		\sqrt{\frac{2K+1}{2J+1}}\,
		C^{\,JJ }_{JJ,\,K0}.
		\label{eq:app-north-projector-expansion}
	\end{equation}
	Rotating Eq.~\eqref{eq:app-north-projector-expansion} and using the expression of the Wigner $D$-matrices~\cite{Varshalovich:1988aa} we get
	\begin{equation}
		P^{(J)}(\n) =
		\sum_{K=0}^{2J}\sum_{q=-K}^{K}
		a_{JK}\,Y^*_{Kq}(\n)\,T^{(J)}_{Kq},
	\end{equation}
	where $a_{JK}$ are defined in \eqref{eq:ajL-main}.

	Using the irreducible tensors, the spin  generators can be expressed as 
	\begin{equation}
		J_q
		= \sqrt{\frac{J(J+1)d}{3}}\,
		T^{(J)}_{1q}.
	\end{equation}
	Together with Eq.~\eqref{eq:tensor-K}, this gives
	\begin{equation}
		K_{J_q}(\n)=
		(J+1)n_q .
		\label{eq:spin-contravariant-symbol}
	\end{equation}
	
	To derive the second moment contravariant symbol, recall the symmetrized second moment
	$ C_{ab} = \tfrac12\{J_a,J_b\}$.  Define the following SO(3) covariant quantity
	\begin{equation}
		I_{ab}
		=
		\frac{2J+1}{4\pi}
		\int_{\sphere}
		n_a n_b\,P^{(J)}(\n)\,d\Omega .
		\label{eq:app-Iab-definition}
	\end{equation}
	By covariance, \(I_{ab}\) must have the form
	\(I_{ab}
	=
	\alpha C_{ab}
	+
	\beta\delta_{ab} \I_{2J+1} \).
	Taking the trace gives
	\begin{equation}
		\frac{1}{3}
		=
		\frac{J(J+1)}{3}\alpha
		+
		\beta .
		\label{eq:app-alpha-beta-trace}
	\end{equation}
	To determine the second coefficient, evaluate the \(zz\) component in the highest-weight state.
	Since
	\begin{equation}
		|\langle JJ|\n\rangle_J|^2
		=
		\left(\frac{1+\cos\theta}{2}\right)^{2J},
	\end{equation}
	one obtains
	\begin{equation}
		\langle JJ|I_{zz}|JJ\rangle=
		\frac{2J+1}{2}
		\int_{-1}^{1}
		x^2\left(\frac{1+x}{2}\right)^{2J}dx=
		\frac{2J^2+J+1}{(J+1)(2J+3)} .
	\end{equation}
	On the other hand,	\(\langle JJ |I_{zz}|JJ\rangle=\alpha J^2+\beta \).
	Solving together with Eq.~\eqref{eq:app-alpha-beta-trace} yields
	\begin{equation}
		\alpha=\frac{2}{(J+1)(2J+3)},
		\qquad
		\beta=\frac{1}{2J+3}.
	\end{equation}
	Therefore
	\begin{equation}
		I_{ab}
		=
		\frac{2}{(J+1)(2J+3)}\,C_{ab}
		+
		\frac{\delta_{ab}}{2J+3}\,\I_{2J+1} .
	\end{equation}
	Inverting this relation gives
	\begin{equation}
		C_{ab}
		=
		\frac{(J+1)(2J+3)}{2}\,I_{ab}
		-
		\frac{J+1}{2}\,\delta_{ab}\I_{2J+1} .
	\end{equation}
	Using Eqs.~\eqref{eq:app-Iab-definition} and the identity resolution
	\begin{equation}
		\I_{2J+1}
		=
		\frac{2J+1}{4\pi}\int_{\sphere}P^{(J)}(\n)\,d\Omega ,
	\end{equation}
	we obtain
	\begin{equation}
		C_{ab} = 	\frac{2J+1}{4\pi}
		\int_{\sphere}
		K_{ab}(\n)\,P^{(J)}(\n)\,d\Omega ,
	\end{equation}
	where
	\begin{equation}
		K_{ab}(\n)
		=
		\frac{(J+1)(2J+3)}{2}\,n_a n_b
		-
		\frac{J+1}{2}\,\delta_{ab}.
		\label{eq:second-moment-contravariant}
	\end{equation}

	\section{Explicit coherent-state decompositions}
	\label{app:explicit-decompositions}
	
	We present two compact examples of the rank theorem.  The first gives a
	particularly simple decomposition of rank-one tensors, while the second
	illustrates the highest nontrivial ranks reconstructed by the regular
	icosahedron.
	
	\subsection{Regular octahedron}
	
	Consider the six vertices of a regular octahedron
	\begin{equation}
		\mathcal O
		=
		\{
		\pm\mathbf e_x,
		\pm\mathbf e_y,
		\pm\mathbf e_z
		\}
		\subset \sphere .
	\end{equation}
	They form a spherical \(3\)-design.  According to Proposition \ref{prop:identity-decomposition}, it resolves the
	identity for every \(J\leq 3/2\)
	\begin{equation}
		\I_{2J+1}
		=
		\frac{2J+1}{6}
		\sum_{\n\in\mathcal O}
		P^{(J)}(\n).
	\end{equation}
	For \(J=1\), it also reconstructs the complete rank-one sector.  Writing
	\begin{equation}
		P_{\pm a}
		=
		P^{(1)}(\pm\mathbf e_a),
		\qquad
		a=x,y,z,
	\end{equation}
	one finds
	\begin{equation}
		P_{+a}-P_{-a}=J_a.
	\end{equation}
	Normalized tensors \(T_a=J_a/\sqrt{2}\), with orthogonality condition \(\Tr(T_a^\dagger T_b)=\delta_{ab},\) can be written using the standard Condon--Shortley convention for spherical
	components~\cite{Varshalovich:1988aa}
	\begin{equation}
		T_{1,0}=T_z,
		\qquad
		T_{1,\pm1}
		=
		\mp\frac{T_x\pm iT_y}{\sqrt{2}}.
		\label{eq:app-rank-one-spherical-cartesian}
	\end{equation}
	Then, any rank-one operator \(A^{(1)}=\sum_{a=x,y,z}c_aT_a,\) can be decomposed as
	\begin{equation}
		A^{(1)}
		=
		\frac{1}{\sqrt{2}}
		\sum_{a=x,y,z}
		c_a\left(P_{+a}-P_{-a}\right).
	\end{equation}

	\subsection{Regular icosahedron}
	
	Recall the golden ratio
	\begin{equation}
		\varphi=\frac{1+\sqrt{5}}{2},
	\end{equation}
	and choose the twelve unit vectors forming a regular icosahedron
	\begin{equation}
		\mathcal I
		=
		\frac{1}{\sqrt{1+\varphi^2}}
		\left\{
		(0,\pm1,\pm\varphi),
		(\pm1,\pm\varphi,0),
		(\pm\varphi,0,\pm1)
		\right\}.
	\end{equation}
	They form a spherical
	\(5\)-design 
	\(\mathcal I=\{\n_k\}_{k=1}^{12}\).  Writing
	\(P_k^{(J)}=P^{(J)}(\n_k)\) in Eq.\eqref{eq:sc-identity-resolution}, one has the following identity resolution
	\begin{equation}
		\I_{2J+1}
		=
		\frac{2J+1}{12}
		\sum_{k=1}^{12}P_k^{(J)},
	\end{equation}
	for \(J\leq5/2\). The same point set reconstructs the complete rank-one sector for
	\(J\leq2\) and the complete rank-two sector for \(J\leq3/2\), according to Theorem~\ref{thm:rank-r-quadrature}.
	
	At the maximal spin \(J=2\) for rank-one reconstruction,
	Eqs.~\eqref{eq:exp} and~\eqref{eq:spin-contravariant-symbol} give
	\begin{equation}
		J_a
		=
		\frac{5}{4}
		\sum_{k=1}^{12}
		n_{k,a}P_k^{(2)},
		\qquad
		a=x,y,z.
	\end{equation}
	Hence every rank-one spherical tensor follows immediately by taking the
	corresponding linear combinations of \(J_x,J_y,J_z\).
	
	For \(J=3/2\), define the traceless Cartesian rank-two tensor
	\begin{equation}
		Q_{ab}
		=
		\frac{1}{2}\{J_a,J_b\}
		-
		\frac{J(J+1)}{3}\delta_{ab}\I_{2J+1}
		=
		\frac{1}{2}\{J_a,J_b\}
		-
		\frac{5}{4}\delta_{ab}\I_4 .
	\end{equation}
	This is the symmetric traceless part of the Cartesian tensor
	\(J_aJ_b\), and therefore transforms as an irreducible tensor of  rank two~\cite{Varshalovich:1988aa}.
	Using the contravariant symbol given by Eq.~\eqref{eq:second-moment-contravariant} one can derive exact icosahedral decomposition
	\begin{equation}
		Q_{ab}
		=
		\frac{5}{2}
		\sum_{k=1}^{12}
		\left(
		n_{k,a}n_{k,b}
		-
		\frac{1}{3}\delta_{ab}
		\right)
		P_k^{(3/2)}.
	\end{equation}
	The five spherical components \(T_{2q}\) are obtained from the standard linear combinations of the Cartesian tensors \(Q_{ab}\). 
	With the Condon--Shortley convention and  normalization \(\Tr(T_{2q}^{\dagger}T_{2q'})=\delta_{qq'}\), the spherical
	components are
	\begin{align}
		T_{2,0}
		&=
		\frac{Q_{zz}}{2},
		&
		T_{2,\pm1}
		&=
		\mp\frac{Q_{xz}\pm iQ_{yz}}{\sqrt{6}},\nonumber
		\\
		T_{2,\pm2}
		&=
		\frac{Q_{xx}-Q_{yy}\pm2iQ_{xy}}{2\sqrt{6}}.
	\end{align}
	
	\begin{acknowledgments}
		M.R. and K.{\.Z}. gratefully acknowledge financial support from the European Union under ERC Advanced Grant TAtypic, project number 101142236. 	A.K. is grateful for support from the Mexican Consejo Nacional de Humanidades, Ciencias y Tecnologías (Grant No. CBF2023-2024-50). L.L.S.-S. thanks the Spanish Agencia Estatal de Investigación (Grant PID2025-172975NB-I00) for its support.

        GPT Astra 6.2 was used for limited assistance with derivations, plotting code, and editorial work. All outputs were reviewed and verified by the authors, who take full responsibility for the results.

	\end{acknowledgments}
	
	\bibliographystyle{jsty3-author}
	\bibliography{biblio}

@article{BaganBaigMunozTapia2001,
  author  = {Bagan, E. and Baig, M. and Mu{\~n}oz-Tapia, R.},
  title   = {Aligning Reference Frames with Quantum States},
  journal = {Phys. Rev. Lett.},
  volume  = {87},
  pages   = {257903},
  year    = {2001},
  doi     = {10.1103/PhysRevLett.87.257903}
}

@article{HeissWeigert2000,
  author  = {Heiss, Stephan and Weigert, Stefan},
  title   = {Discrete {Moyal}-type representations for a spin},
  journal = {Phys. Rev. A},
  volume  = {63},
  pages   = {012105},
  year    = {2000},
  doi     = {10.1103/PhysRevA.63.012105}
}

@book{Perelomov:1986aa,
  author     = {Perelomov, A. M.},
  title      = {Generalized Coherent States and Their Applications},
  publisher  = {Springer},
  address    = {Berlin},
  year       = {1986},
  doi        = {10.1007/978-3-642-61629-7},
  ISBN = {0-387-15912-6}
}

@book{Gazeau:2009aa,
  author     = {Gazeau, J.-P.},
  title      = {Coherent States in Quantum Physics},
  publisher  = {Wiley-VCH},
  address    = {Weinheim},
  year       = {2009},
  ISBN = {978-3-527-62829-2},
  doi        = {10.1002/9783527628285}
}

@book{Robert:2021aa,
  author     = {D. Robert and M. Combescure},
  title      = {Coherent States and Applications in Mathematical Physics},
  edition    = {Second},
  publisher  = {Springer},
  address    = {Cham},
  year       = {2021},
  ISBN = {978-3-030-70844-3},
  doi        = {10.1007/978-3-030-70845-0}
}

@book{Kam:2023aa,
  author     = {C.-F. Kam and W.-M. Zhang and D.-H. Feng},
  title     = {Coherent States: New Insights into Quantum Mechanics with Applications},
  publisher  = {Springer},
  address    = {Cham},
  year       = {2023},
  ISBN = {978-3-031-20765-5},
  doi        = {10.1007/978-3-031-20766-2}
}

@book{Klauder:1985aa,
  editor     = {Klauder, John R. and Skagerstam, Bo-Sture},
  title      = {Coherent States: Applications in Physics and Mathematical Physics},
  publisher  = {World Scientific},
  address    = {Singapore},
  year       = {1985},
  ISBN = {978-981-4590-88-4},
  doi        = {10.1142/0096}
}

@book{Schroeck:1996fv,
  author     = {F. E. Schroeck},
  title      = {Quantum Mechanics on Phase Space},
  publisher  = {Kluwer Academic Publishers},
  address    = {Dordrecht},
  year       = {1996},
  ISBN = {978-94-017-2830-0},
  doi        = {10.1007/978-94-017-2830-0}
}

@book{Schleich:2001hc,
  author     = {W. P. Schleich},
  title      = {Quantum Optics in Phase Space},
  publisher  = {Wiley-VCH},
  address    = {Berlin},
  year       = {2001},
  ISBN = {9783527294350},
  doi        = {10.1002/3527602976}
}

@book{QMPS:2005mi,
  editor     = {C. K. Zachos and D. B. Fairlie and T. L. Curtright},
  title      = {Quantum Mechanics in Phase Space: An Overview with Selected Papers},
  publisher  = {World Scientific},
  address    = {Singapore},
  year       = {2005},
  ISBN = {978-981-4485-87-6},
  doi        = {10.1142/5287}
}

@article{Janssen:1981aa,
  author     = {Janssen, A. J. E. M.},
  title      = {Weighted {Wigner} distributions vanishing on lattices},
  journal    = {J. Math. Anal. Appl.},
  volume     = {80},
  pages      = {156--167},
  year       = {1981},
  doi        = {10.1016/0022-247X(81)90099-8}
}

@book{Ueberhuber:1997aa,
  author     = {C. W. Ueberhuber},
  title      = {Numerical Computation 2: Methods, Software, and Analysis},
  publisher  = {Springer},
  address    = {Berlin},
  year       = {1997},
  doi        = {10.1007/978-3-642-59109-9}, 
  ISBN = {9783642591099},
  url = {https://link.springer.com/book/9783540620570}
}

@article{Delsarte:1977aa,
  author     = {Delsarte, P. and Goethals, J. M. and Seidel, J. J.},
  title      = {Spherical codes and designs},
  journal    = {Geom. Dedicata},
  volume     = {6},
  pages      = {363--388},
  year       = {1977},
  doi        = {10.1007/bf03187604}
}

@incollection{Goethals:1981aa,
  author     = {Goethals, J. M. and Seidel, J. J.},
  title      = {Cubature Formulae, Polytopes, and Spherical Designs},
  booktitle  = {The Geometric Vein},
  editor     = {Davis, Chandler and Gr{\"u}nbaum, Branko and Sherk, F. A.},
  pages      = {203--218},
  publisher  = {Springer},
  address    = {New York},
  year       = {1981},
  doi        = {10.1007/978-1-4612-5648-9_13}
}

@book{Colbourn:2010aa,
  editor    = {Colbourn, Charles J. and Dinitz, Jeffrey H.},
  title     = {Handbook of Combinatorial Designs},
  edition   = {2},
  publisher = {Chapman \& Hall/CRC},
  address   = {Boca Raton, FL},
   ISBN = {9780429138485},
  year      = {2007}
}

@book{Neumann:1932aa,
  author     = {von Neumann, John},
  title      = {Mathematische Grundlagen der Quantenmechanik},
  publisher  = {Springer},
  address    = {Berlin},
  year       = {1932}
}

@article{Perelomov:1971aa,
  author     = {Perelomov, A. M.},
  title      = {On the completeness of a system of coherent states},
  journal    = {Theor. Math. Phys.},
  volume     = {6},
  pages      = {156--164},
  year       = {1971},
  doi        = {10.1007/BF01036577}
}

@article{Bargmann:1971aa,
  author     = {Bargmann, V. and Butera, P. and Girardello, L. and Klauder, John R.},
  title      = {On the completeness of the coherent states},
  journal    = {Rep. Math. Phys.},
  volume     = {2},
  pages      = {221--228},
  year       = {1971},
  doi        = {10.1016/0034-4877(71)90006-1}
}

@article{Bacry:1975aa,
  author     = {Bacry, H. and Grossmann, A. and Zak, J.},
  title      = {Proof of completeness of lattice states in the $kq$ representation},
  journal    = {Phys. Rev. B},
  volume     = {12},
  pages      = {1118--1120},
  year       = {1975},
  doi        = {10.1103/PhysRevB.12.1118}
}

@article{Boon:1983aa,
  author     = {Boon, M. and Zak, J. and Zucker, I. J.},
  title      = {Rational von {Neumann} lattices},
  journal    = {J. Math. Phys.},
  volume     = {24},
  pages      = {316--323},
  year       = {1983},
  doi        = {10.1063/1.525682}
}

@article{Gottesman:2001aa,
  author     = {Gottesman, Daniel and Kitaev, Alexei and Preskill, John},
  title      = {Encoding a qubit in an oscillator},
  journal    = {Phys. Rev. A},
  volume     = {64},
  pages      = {012310},
  year       = {2001},
  doi        = {10.1103/physreva.64.012310}
}

@article{Menicucci:2014aa,
  author     = {Menicucci, Nicolas C.},
  title      = {Fault-tolerant measurement-based quantum computing with continuous-variable cluster states},
  journal    = {Phys. Rev. Lett.},
  volume     = {112},
  pages      = {120504},
  year       = {2014},
  doi        = {10.1103/PhysRevLett.112.120504}
}

@article{Duivenvoorden:2017aa,
  author     = {Duivenvoorden, Kasper and Terhal, Barbara M. and Weigand, Daniel},
  title      = {Single-mode displacement sensor},
  journal    = {Phys. Rev. A},
  volume     = {95},
  pages      = {012305},
  year       = {2017},
  doi        = {10.1103/PhysRevA.95.012305}
}

@article{Fukui:2021aa,
  author     = {Fukui, Kosuke and Alexander, Rafael N. and van Loock, Peter},
  title      = {All-optical long-distance quantum communication with {G}ottesman-{K}itaev-{P}reskill qubits},
  journal    = {Phys. Rev. Res.},
  volume     = {3},
  pages      = {033118},
  year       = {2021},
  doi        = {10.1103/PhysRevResearch.3.033118}
}

@article{Conrad:2022aa,
  author     = {Conrad, Jonathan and Eisert, Jens and Arzani, Francesco},
  title      = {Gottesman-{K}itaev-{P}reskill codes: {A} lattice perspective},
  journal    = {{Quantum}},
  volume     = {6},
  pages      = {648},
  year       = {2022},
  doi        = {10.22331/q-2022-02-10-648}
}

@article{Atkins:1971aa,
  author     = {Atkins, P. W. and Dobson, J. C.},
  title      = {Angular momentum coherent states},
  journal    = {Proc. R. Soc. A},
  volume     = {321},
  pages      = {321--340},
  year       = {1971},
  doi        = {10.1098/rspa.1971.0035}
}

@article{Radcliffe:1971aa,
  author     = {Radcliffe, J. M.},
  title      = {Some properties of coherent spin states},
  journal    = {J. Phys. A: Gen. Phys.},
  volume     = {4},
  pages      = {313--323},
  year       = {1971},
  doi        = {10.1088/0305-4470/4/3/009}
}

@article{Arecchi:1972aa,
  author     = {Arecchi, F. T. and Courtens, Eric and Gilmore, Robert and Thomas, Harry},
  title      = {Atomic coherent states in quantum optics},
  journal    = {Phys. Rev. A},
  volume     = {6},
  pages      = {2211--2237},
  year       = {1972},
  doi        = {10.1103/physreva.6.2211}
}

@article{Kitagawa:1993aa,
  author     = {Kitagawa, Masahiro and Ueda, Masahito},
  title      = {Squeezed spin states},
  journal    = {Phys. Rev. A},
  volume     = {47},
  pages      = {5138--5143},
  year       = {1993},
  doi        = {10.1103/PhysRevA.47.5138}
}

@article{Dammeier:2015aa,
  author     = {Dammeier, Lars and Schwonnek, Ren{\'e} and Werner, Reinhard F},
  title      = {Uncertainty relations for angular momentum},
  journal    = {New J. Phys.},
  volume     = {17},
  pages      = {093046},
  year       = {2015},
  doi        = {10.1088/1367-2630/17/9/093046}
}

@article{Shabbir:2016aa,
  author     = {Shabbir, Saroosh and Bj\"ork, Gunnar},
  title      = {{SU(2)} uncertainty limits},
  journal    = {Phys. Rev. A},
  volume     = {93},
  pages      = {052101},
  year       = {2016},
  doi        = {10.1103/PhysRevA.93.052101}
}

@article{Goldberg:2020aa,
  author     = {Goldberg, Aaron Z. and Klimov, Andrei B. and Grassl, Markus and Leuchs, Gerd and S{\'a}nchez-Soto, Luis L.},
  title      = {Extremal quantum states},
  journal    = {AVS Quantum Sci.},
  volume     = {2},
  pages      = {044701},
  year       = {2020},
  doi        = {10.1116/5.0025819}
}

@inproceedings{Ambainis:2007aa,
  author     = {Ambainis, Andris and Emerson, Joseph},
  title      = {Quantum $t$-designs: $t$-wise Independence in the Quantum World},
  booktitle  = {Twenty-Second Annual IEEE Conference on Computational Complexity (CCC'07)},
  pages      = {129--140},
  year       = {2007},
  doi        = {10.1109/CCC.2007.26}
}

@article{Saff:1997aa,
  author     = {Saff, Edward B and Kuijlaars, Amo B J},
  title      = {Distributing many points on a sphere},
  journal    = {Math. Intelligencer},
  volume     = {19},
  pages      = {5--11},
  year       = {1997},
  doi        = {10.1007/BF03024331}
}

@article{Conway:1996ys,
  author     = {Conway, John H. and Hardin, Ronald H. and Sloane, Neil J. A.},
  title      = {Packing Lines, Planes, etc.: Packings in {G}rassmannian Spaces},
  journal    = {Exp. Math.},
  volume     = {5},
  pages      = {139--159},
  year       = {1996},
  doi = {10.1080/10586458.1996.10504585},
  url        = {http://eudml.org/doc/226058}
}

@article{Viazovska:2017aa,
  author     = {Viazovska, Maryna S.},
  title      = {The sphere packing problem in dimension 8},
  journal    = {Ann. Math.},
  volume     = {185},
  pages      = {991--1015},
  year       = {2017},
  doi = {10.4007/annals.2017.185.3.7},
  url        = {http://www.jstor.org/stable/26395747}
}

@article{Berry:1977aa,
  author     = {Berry, H. G. and Gabrielse, G. and Livingston, A. E.},
  title      = {Measurement of the {Stokes} parameters of light},
  journal    = {Appl. Opt.},
  volume     = {16},
  pages      = {3200--3205},
  year       = {1977},
  doi        = {10.1364/AO.16.003200}
}

@article{Azzam:1985aa,
  author     = {Azzam, R. M. A.},
  title      = {Arrangement of four photodetectors for measuring the state of polarization of light},
  journal    = {Opt. Lett.},
  volume     = {10},
  pages      = {309--311},
  year       = {1985},
  doi        = {10.1364/OL.10.000309}
}

@article{Goldberg:2020ac,
  author     = {A. Z. Goldberg and P. de la Hoz and G. Bj\"{o}rk and A. B. Klimov and M. Grassl and G. Leuchs and L. L. S\'{a}nchez-Soto},
  title      = {Quantum concepts in optical polarization},
  journal    = {Adv. Opt. Photon.},
  volume     = {13},
  pages      = {1--73},
  year       = {2021},
  doi        = {10.1364/AOP.404175}
}

@article{Wasilewski:2010aa,
  author     = {Wasilewski, W. and Jensen, K. and Krauter, H. and Renema, J. J. and Balabas, M. V. and Polzik, E. S.},
  title      = {Quantum Noise Limited and Entanglement-Assisted Magnetometry},
  journal    = {Phys. Rev. Lett.},
  volume     = {104},
  pages      = {133601},
  year       = {2010},
  doi        = {10.1103/PhysRevLett.104.133601}
}

@article{Behbood:2013aa,
  author     = {Behbood, N. and Martin Ciurana, F. and Colangelo, G. and Napolitano, M. and Mitchell, M. W. and Sewell, R. J.},
  title      = {Real-time vector field tracking with a cold-atom magnetometer},
  journal    = {Appl. Phys. Lett.},
  volume     = {102},
  pages      = {173504},
  year       = {2013},
  doi        = {10.1063/1.4803684}
}

@article{Renes:2004aa,
  author     = {Renes, Joseph M. and Blume-Kohout, Robin and Scott, A. J. and Caves, Carlton M.},
  title      = {Symmetric informationally complete quantum measurements},
  journal    = {J. Math. Phys.},
  volume     = {45},
  pages      = {2171--2180},
  year       = {2004},
  doi        = {10.1063/1.1737053}
}

@inproceedings{Klappenecker:2005aa,
  author     = {Klappenecker, A. and R{\"o}tteler, M.},
  title      = {Mutually unbiased bases are complex projective 2-designs},
  booktitle  = {IEEE Int. Symp. Inf. Theory - Proc.},
  pages      = {1740--1744},
  year       = {2005},
  doi        = {10.1109/ISIT.2005.1523643}
}

@article{Cieslinski:2024aa,
  author     = {Pawe{\l} Cie{\'s}li{\'n}ski and Satoya Imai and Jan Dziewior and Otfried G{\"u}hne and Lukas Knips and Wies{\l}aw Laskowski and Jasmin Meinecke and Tomasz Paterek and Tam{\'a}s V{\'e}rtesi},
  title      = {Analysing quantum systems with randomised measurements},
  journal    = {Phys. Rep.},
  volume     = {1095},
  pages      = {1--48},
  year       = {2024},
  doi        = {10.1016/j.physrep.2024.09.009}
}

@article{Blume:2014aa,
  author     = {Blume-Kohout, Robin and Turner, Peter S},
  title      = {The curious nonexistence of {Gaussian} 2-designs},
  journal    = {Commun. Math. Phys.},
  volume     = {326},
  pages      = {755--771},
  year       = {2014},
  doi        = {10.1007/s00220-014-1894-3}
}

@article{Bannai:2009aa,
  author     = {Bannai, Eiichi and Bannai, Etsuko},
  title      = {A survey on spherical designs and algebraic combinatorics on spheres},
  journal    = {Eur. J. Comb.},
  volume     = {30},
  pages      = {1392--1425},
  year       = {2009},
  doi        = {10.1016/j.ejc.2008.11.007}
}

@incollection{Womersley:2018aa,
  author     = {Womersley, Robert S.},
  title      = {Efficient Spherical Designs with Good Geometric Properties},
  booktitle  = {Contemporary Computational Mathematics: A Celebration of the 80th Birthday of Ian Sloan},
  editor     = {Dick, J. and Kuo, F. and Wo{\'z}niakowski, H.},
  pages      = {1243--1285},
  publisher  = {Springer International Publishing},
  address    = {Cham},
  year       = {2018},
  doi        = {10.1007/978-3-319-72456-0_57}
}

@article{Seymour:1984aa,
  author     = {Seymour, P. D and Zaslavsky, Thomas},
  title      = {Averaging sets: A generalization of mean values and spherical designs},
  journal    = {Adv. Math.},
  volume     = {52},
  pages      = {213--240},
  year       = {1984},
  doi        = {10.1016/0001-8708(84)90022-7}
}

@article{Hardin:1996aa,
  author     = {Hardin, R. H. and Sloane, N. J. A.},
  title      = {Mc{L}aren's improved snub cube and other new spherical designs in three dimensions},
  journal    = {Discrete Comput. Geom.},
  volume     = {15},
  pages      = {429--441},
  year       = {1996},
  doi        = {10.1007/bf02711518}
}

@misc{SloaneRepository,
  author     = {Hardin, R. H. and Sloane, N. J. A},
  title      = {Spherical Designs Library},
  year       = {},
  url        = {https://neilsloane.com/sphdesigns/},
  note       = {Accessed 26 July 2026}
}

@article{Sloan:2009aa,
  author     = {Sloan, I. H. and Womersley, R. S.},
  title      = {A variational characterisation of spherical designs},
  journal    = {J. Approx. Theory},
  volume     = {159},
  pages      = {308--318},
  year       = {2009},
  doi        = {10.1016/j.jat.2009.02.014}
}

@article{Bondarenko:2013aa,
  author     = {Bondarenko, Andriy and Radchenko, Danylo and Viazovska, Maryna},
  title      = {Optimal asymptotic bounds for spherical designs},
  journal    = {Ann. Math.},
  volume     = {178},
  pages      = {443--452},
  year       = {2013},
  doi        = {10.4007/annals.2013.178.2.2}
}

@book{Fano:1959ly,
  author     = {U. Fano and G. Racah},
  title      = {Irreducible Tensorial Sets},
  publisher  = {Academic Press},
  address    = {New York},
  year       = {1959},
  ISBN = {978-1-4832-3059-7},
  url        = {https://www.sciencedirect.com/bookseries/pure-and-applied-physics/vol/4/suppl/C}
}

@book{Blum:1981rb,
  author     = {Blum, Karl},
  title      = {Density Matrix Theory and Applications},
  publisher  = {Springer Berlin},
  address    = {Heidelberg},
  edition    = {3},
  year       = {2012},
  ISBN = {978-3-642-20561-3},
  doi        = {10.1007/978-3-642-20561-3}
}

@book{Varshalovich:1988aa,
  author     = {Varshalovich, D A and Moskalev, A N and Khersonskii, V K},
  title      = {Quantum Theory of Angular Momentum},
  publisher  = {World Scientific},
  address    = {Singapore},
  year       = {1988},
  ISBN = {978-9971-5-0996-5},
  doi        = {10.1142/0270}
}

@book{Peres:2002oz,
  author     = {Peres, Asher},
  title      = {Quantum {T}heory: {C}oncepts and {M}ethods},
  publisher  = {Kluwer Academic Publishers},
  address    = {Dordrecht},
  year       = {1993},
  ISBN = {978-0-7923-2549-9},
  doi        = {10.1007/0-306-47120-5}
}

@article{Stratonovich:1957aa,
  author     = {Stratonovich, R. L.},
  title      = {On distributions in representation space},
  journal    = {Soviet Phys. JETP},
  volume     = {4},
  pages      = {891--898},
  year       = {1957},
  url        = {https://jetp.ras.ru/cgi-bin/dn/e_004_06_0891.pdf}
}

@article{Berezin:1975mw,
  author     = {F. A. Berezin},
  title      = {General concept of quantization},
  journal    = {Commun. Math. Phys.},
  volume     = {40},
  pages      = {153--174},
  year       = {1975},
  doi        = {10.1007/BF01609397}
}

@article{Varilly:1989aa,
  author     = {V{\'a}rilly, Joseph C and Gracia-Bond{\'\i}a, Jos{\'e} M.},
  title      = {The {Moyal} representation for spin},
  journal    = {Ann. Phys.},
  volume     = {190},
  pages      = {107--148},
  year       = {1989},
  doi        = {10.1016/0003-4916(89)90262-5}
}

@article{Dowling:1994sw,
  author     = {J. P. Dowling and G. S. Agarwal and W. P. Schleich},
  title      = {Wigner distribution of a general angular-momentum state: {A}pplications to a collection of two-level atoms.},
  journal    = {Phys. Rev. A},
  volume     = {49},
  pages      = {4101--4109},
  year       = {1994},
  doi        = {10.1103/PhysRevA.49.4101}
}

@article{Brif:1999aa,
  author     = {Brif, C. and Mann, A.},
  title      = {Phase-space formulation of quantum mechanics and quantum-state reconstruction for physical systems with {Lie}-group symmetries},
  journal    = {Phys. Rev. A},
  volume     = {59},
  pages      = {971--987},
  year       = {1999},
  doi        = {10.1103/physreva.59.971}
}

@incollection{GoethalsSeidel1979,
  author     = {Goethals, J. M. and Seidel, J. J.},
  title      = {Spherical Designs},
  booktitle  = {Relations Between Combinatorics and Other Parts of Mathematics},
  editor     = {Ray-Chaudhuri, D. K.},
  series     = {Proceedings of Symposia in Pure Mathematics},
  volume     = {34},
  pages      = {255--272},
  publisher  = {American Mathematical Society},
  address    = {Providence, RI},
  year       = {1979},
  url        = {https://bookstore.ams.org/pspum/34}
}

@unpublished{Rudzinski2026,
  author     = {V{\'a}zquez Mota, I. and Rudzi{\'n}ski, M. and Chryssomalakos, C. and {\.Z}yczkowski, K.},
  title      = {Coloring the quantum stars: diverse applications of a single {M}ajorana constellation},
  year       = {to appear, 2026},
  note       = {In preparation}
}

@article{Goyenecheetal2015,
  author     = {Goyeneche, Dardo and Alsina, Daniel and Latorre, Jos\'e I. and Riera, Arnau and {\.Z}yczkowski, Karol},
  title      = {Absolutely maximally entangled states, combinatorial designs, and multiunitary matrices},
  journal    = {Phys. Rev. A},
  volume     = {92},
  pages      = {032316},
  year       = {2015},
  doi        = {10.1103/PhysRevA.92.032316}
}

@article{Goyenecheetal2018,
  author     = {Goyeneche, Dardo and Raissi, Zahra and Di Martino, Sara and {\.Z}yczkowski, Karol},
  title      = {Entanglement and quantum combinatorial designs},
  journal    = {Phys. Rev. A},
  volume     = {97},
  pages      = {062326},
  year       = {2018},
  doi        = {10.1103/PhysRevA.97.062326}
}

@article{Ruhl1982,
  author     = {R{\"u}hl, Werner},
  title      = {A boson representation for {SU(N)} lattice gauge theories},
  journal    = {Commun. Math. Phys.},
  volume     = {83},
  pages      = {455--468},
  year       = {1982},
  doi        = {10.1007/BF01208711}
}

@article{AmietWeigert1999,
  author  = {Amiet, Jean-Pierre and Weigert, Stefan},
  title   = {Coherent states and the reconstruction of pure spin states},
  journal = {J. Opt. B: Quantum Semiclassical Opt.},
  volume  = {1},
  pages   = {L5--L8},
  year    = {1999},
  doi     = {10.1088/1464-4266/1/5/101}
}

@article{Weigert1999,
  author        = {Weigert, Stefan},
  title         = {A discrete phase-space calculus for quantum spins based on a reconstruction method using coherent states},
  journal       = {Acta Phys. Slovaca},
  volume        = {49},
  pages         = {613--620},
  year          = {1999},
  eprint        = {quant-ph/9904095},
  archivePrefix = {arXiv},
  primaryClass  = {quant-ph},
  doi           = {10.48550/arXiv.quant-ph/9904095},
  url           = {https://arxiv.org/abs/quant-ph/9904095}
}

@article{AmietWeigert2000,
  author  = {Amiet, Jean-Pierre and Weigert, Stefan},
  title   = {Discrete {$Q$}- and {$P$}-symbols for spin {$s$}},
  journal = {J. Opt. B: Quantum Semiclassical Opt.},
  volume  = {2},
  pages   = {118--121},
  year    = {2000},
  doi     = {10.1088/1464-4266/2/2/309}
}

@article{Weigert2004,
  author  = {Weigert, Stefan},
  title   = {Expanding {H}ermitian operators in a basis of projectors on coherent spin states},
  journal = {J. Opt. B: Quantum Semiclassical Opt.},
  volume  = {6},
  pages   = {489--490},
  year    = {2004},
  doi     = {10.1088/1464-4266/6/12/001}
}

@article{IblisdirRoland2006,
  author  = {Iblisdir, Sofyan and Roland, J{\'e}r{\'e}mie},
  title   = {Optimal finite measurements and {Gauss} quadratures},
  journal = {Phys. Lett. A},
  volume  = {358},
  pages   = {368--372},
  year    = {2006},
  doi     = {10.1016/j.physleta.2006.05.045}
}

@article{Dariano2004,
 author = {D'Ariano, G. M. and Perinotti, P. and Sacchi, M. F.},
 title = {Informationally complete measurements and groups representation},
 journal = {Journal of Optics B: Quantum and Semiclassical Optics},
 volume = {6}, pages = {S487--S491}, year = {2004},
 doi = {10.1088/1464-4266/6/6/005},
 eprint = {quant-ph/0310013}, archivePrefix = {arXiv}
}

@article{Scott2006,
 author = {Scott, A. J.},
 title = {Tight informationally complete quantum measurements},
 journal = {Journal of Physics A: Mathematical and General},
 volume = {39}, pages = {13507--13530}, year = {2006},
 doi = {10.1088/0305-4470/39/43/009},
 eprint = {quant-ph/0604049}, archivePrefix = {arXiv}
}

@article{Latorre1998,
  title = {Minimal Optimal Generalized Quantum Measurements},
  author = {Latorre, J. I. and Pascual, P. and Tarrach, R.},
  journal = {Phys. Rev. Lett.},
  volume = {81},
  issue = {7},
  pages = {1351--1354},
  numpages = {0},
  year = {1998},
  month = {Aug},
  publisher = {American Physical Society},
  doi = {10.1103/PhysRevLett.81.1351},
  url = {https://link.aps.org/doi/10.1103/PhysRevLett.81.1351}
}

@article{HuTomamichel2024,
  title = {Fundamental limits on quantum cloning from the no-signaling principle},
  author = {Hu, Yanglin and Tomamichel, Marco},
  journal = {Phys. Rev. A},
  volume = {109},
  issue = {2},
  pages = {022221},
  numpages = {18},
  year = {2024},
  month = {Feb},
  publisher = {American Physical Society},
  doi = {10.1103/PhysRevA.109.022221},
  url = {https://link.aps.org/doi/10.1103/PhysRevA.109.022221}
}

\end{document}